\documentclass[a4paper,USenglish]{article}
\usepackage[utf8]{inputenc}
\usepackage[hidelinks]{hyperref}
\usepackage{amsmath,amssymb,amsthm,thmtools,mathtools,xspace,cleveref}

\usepackage{tikz}
\usetikzlibrary{arrows,chains,matrix,positioning,scopes,decorations,patterns,decorations.pathreplacing}
\tikzstyle{vertex}=[circle, draw, inner sep=0pt, minimum width=10pt]

\declaretheorem[numberwithin=section]{theorem}
\declaretheorem[sibling=theorem]{lemma}

\declaretheorem[sibling=theorem]{definition}

\declaretheorem{observation}
\declaretheorem[sibling=theorem]{proposition}
\declaretheorem[numberwithin=section]{remark}
\declaretheorem[sibling=theorem]{corollary}
\declaretheorem[sibling=theorem]{example}

\newcommand{\DP}{\mathsf{P}}

\newcommand{\sharpP}{\mathsf{\#P}}

\newcommand{\FPT}{\mathsf{FPT}}
\newcommand{\W}[1]{\mathsf{W}[#1]}
\newcommand{\sharpFPT}{\mathsf{\#FPT}}
\newcommand{\sharpW}[1]{\mathsf{\#W}[#1]}

\newcommand{\VP}{\mathsf{VP}}
\newcommand{\VNP}{\mathsf{VNP}}

\newcommand{\VFPT}{\mathsf{VFPT}}
\newcommand{\VW}[1]{\mathsf{VW}[#1]}
\newcommand{\VFPTd}{\mathsf{VFPT_\mathit{deg}}}
\newcommand{\VWd}[1]{\mathsf{VW_\mathit{deg}}[#1]}

\newcommand{\ones}[2]{\genfrac{\langle}{\rangle}{0pt}{}{#1}{#2}}

\newcommand{\Nset}{\mathbb{N}}

\newcommand{\hp}[1]{H_{#1}}
\newcommand{\spc}[1]{S_{#1}}

\newcommand{\lep}{\le_p}
\newcommand{\lec}{\le_c}
\newcommand{\les}{\le_s}
\newcommand{\lefp}{\le^\mathit{fpt}_p}
\newcommand{\lefc}{\le^\mathit{fpt}_c}
\newcommand{\lefs}{\le^\mathit{fpt}_s}

\newcommand{\poly}{\operatorname{poly}}

\newcommand{\VC}{\operatorname{VC}}
\newcommand{\per}{\operatorname{per}}
\newcommand{\Cl}{\operatorname{Clique}}

\newcommand{\cG}{\mathcal{G}}
\newcommand{\rper}{\operatorname{rper}}
\newcommand{\sVC}{\operatorname{VC}^s}

\newcommand{\cB}{\mathcal{B}}

\newcommand{\Psub}{\ensuremath{\texttt{PartitionedSub}}}
\newcommand{\ksparseperm}{\ensuremath{(\per_{\texttt{sparse},k,c})}\xspace}

\title{Parameterized Valiant's classes}
\author{Markus Bl\"aser\footnote{Computational Complexity Department, Saarlnad University, Germany, {mblaeser@cs.uni-saarland.de}} \and Christian Engels\footnote{CSE Department, IIT Bombay, India {christian@cse.iitb.ac.in}}}

\begin{document}

\maketitle

\begin{abstract}
We define a theory of parameterized algebraic complexity classes
in analogy to parameterized Boolean counting classes. We define the classes
$\VFPT$ and $\VW t$, which mirror the Boolean counting classes $\sharpFPT$
and $\sharpW t$, and define appropriate reductions and completeness notions.
Our main contribution is the $\VW 1$-completeness proof of the parameterized
clique family. This proof is far more complicated than in the Boolean world.
It requires some new concepts like composition theorems for bounded exponential
sums and Boolean-arithmetic formulas.
In addition, we also look at two polynomials linked to the permanent with vastly different parameterized complexity.
\end{abstract}

\section{Introduction}
When Valiant invented the theory of computational counting and $\sharpP$-completeness, he also defined
an algebraic model for computing families of polynomials~\cite{DBLP:conf/stoc/Valiant79a}.
This was very natural, since many (Boolean) counting problems are evaluations of polynomials:
Counting perfect matchings in bipartite graphs is the same as evaluating the permanent
at the adjacency matrix, counting Hamiltonian tours in directed graphs is the same as
evaluating the Hamiltonian cycle polynomial at the adjacency matrix, etc.
There is a fruitful interplay between the Boolean and the algebraic world:
algebraic methods like interpolation
can be used to design counting algorithms as well as for proving hardness results.
Proving lower bounds might be easier in the algebraic world and then we can use
transfer theorems from the algebraic world to the Boolean world~\cite{DBLP:journals/tcs/Burgisser00}.

Parameterized counting complexity has been very successful in the recent years,
see for instance~\cite{DBLP:Thore,DBLP:conf/focs/CurticapeanM14,DBLP:journals/corr/CurticapeanDFGL17,
DBLP:stoc:radu,DBLP:journals/combinatorica/JerrumM17,DBLP:Roth}.
Parameterized complexity provides a more fine-grained
view on $\sharpP$-complete problems. There are problems like counting vertex covers
of size $k$, which are fixed-parameter tractable, and others,
which are presumably harder, like the problem of counting cliques of size $k$.
Beside the classes $\VP$ and $\VNP$ (and subclasses of them), which correspond
to time bounded computation in the Boolean world, there have also been definitions of
algebraic classes that correspond to space bounded Boolean computation, see~\cite{DBLP:conf/stacs/KoiranP07, DBLP:conf/mfcs/KoiranP07, DBLP:conf/fct/MahajanR09}.
However, we are not aware of any parameterized classes in the algebraic world despite some algorithmic upper bounds being known, see for instance~\cite{DBLP:journals/dam/CourcelleMR01, DBLP:conf/isaac/FlarupKL07}.

\subsection{Our Contribution}

In this paper, we define a theory of parameterized algebraic complexity classes.
While some of the definitions are rather obvious modifications of the Boolean ones
and some of the basic theorems easily transfer from the Boolean world
to the algebraic world, some concepts have to be modified.
For instance, we cannot use projections to define hardness in general in the
parameterized world, since they can only decrease the degree. On the one hand, one could choose the
degree as a parameter, for instance, when computing the vertex cover polynomial. On the other hand, one could have parameterized families where the degree is always $n$,
like the permanent on bounded genus graphs.
One cannot compare these families with projections, although both of them turn out to
be fixed parameter tractable. We could use c-reductions instead (which are the analogue
of Turing reductions). However, these seem too powerful. We propose some intermediate concept, namely fpt-substitutions:
We may replace the variables of a polynomial by other polynomials that are
computed by small circuits (and not simply constants and variables like in the case of projections).
This mirrors what is done in parsimonious reductions: the input is transformed by a polynomial time computable
function but no post-processing is allowed.

Our main technical contribution is the $\VW 1$-completeness proof of the parameterized
clique polynomial. This proof turns out to be far more complicated than in the
Boolean world, since we are not counting satisfying assignments to Boolean circuits
but we are computing sums over algebraic circuits. First we prove
that one can combine two exponential sums into one sum. While this is very easy
in the case of $\VNP$, it turns out to be quite complicated in the case of $\VW 1$.
Then we prove a normal form for so-called weft 1 circuits, the defining circuits
for $\VW 1$.  We go on with proving that the components consisting of all monomials that
depend on a given number of variables of a polynomial
computed by a weft 1 formula can be written as a bounded exponential sum
over so-called Boolean-arithmetic expressions.
We then show how to reduce such
a sum to the clique problem.

In our final section, we study two polynomials based on cycle covers. The first one consisting of one cycle of length $k$ and all other cycles being self loops. The second is similar, but allows all other cycles to be a certain constant. We prove that the first problem is $\VW{1}$ complete while the second problem is hard for $\VW{t}$ for all $t$.

\subsection{Motivation}
It is an open problem in Boolean parameterized complexity theory, whether $\W t = \W {t+1}$ or
$\sharpW t = \sharpW {t+1}$ implies a collapse of the corresponding hierarchy.
Here algebraic complexity theory might shed a new light on these problems, since algebraic circuits
are much more structured than Boolean ones.

Additionally, we find the difference between the complexity of various polynomials based on the permanent very interesting and worthy of future studies.

We build the basic framework for the study of these questions in this paper.

\section{Valiant's Classes}

We give a brief introduction to Valiant's classes, for further information
we refer the reader to~\cite{Buergisser:00,DBLP:journals/corr/Mahajan13}.
Throughout the whole paper, $K$
will denote the underlying ground field.
An arithmetic circuit $C$ is an acyclic directed graph such that every node has
either indegree $0$ or indegree $2$.
Nodes with indegree $0$ are called input
nodes and they are either labeled with a constant from $K$ or with some variable.
The other nodes are called computation gates, they are either labeled
with $+$ (addition gate) or $*$ (multiplication gate). Every gate computes
a polynomial in the obvious way. There is exactly one gate of outdegree $0$,
the polynomial computed there is the output of $C$. The size of a circuit is the number
of edges in it. The depth of a circuit is length of a longest path from
an input node to the output node. Later, we will also look at circuits in which the computation
gates can have arbitrary fan-in.

The objects that we study will be polynomials over $K$ in variables $X_1,X_2,\dots$
(Occasionally, we will rename these variables to make the presentation more readable.) We will denote $\{ X_1,X_2,\dots\}$ by $X$.
The circuit complexity $C(f)$ of a polynomial $f$ is the size of a smallest
circuit computing $f$.
We call a function $r\colon \Nset \to \Nset$ \emph{p-bounded} if there is a polynomial
$p$ such that $r(n) \le p(n)$ for all $n$.

\begin{definition}\label{def:mb:p-fam}
A sequence of polynomials $(f_n) \in K[X]$ is called a \emph{p-family}
if for all $n$,
\begin{enumerate}
\item $f_n \in K[X_1,\dots,X_{p(n)}]$ for some p-bounded function $p$ and
\item $\deg f_n$ is p-bounded.
\end{enumerate}
\end{definition}

\begin{definition}\label{def:mb:vp}
The class $\VP$ consists of all p-families $(f_n)$ such that
$C(f_n)$ is p-bounded.
\end{definition}

Let $f \in K[X]$ be a polynomial and $s\colon X \to K[X]$ be a mapping that maps indeterminates
to polynomials. Now, $s$ can be extended in a unique way to an algebra endomorphism $K[X] \to K[X]$.
We call $s$ a \emph{substitution}. (Think of the variables being replaced by polynomials.)

\begin{definition}
\begin{enumerate}
\item Let $f,g \in K[X]$. $f$ is called a \emph{projection} of $g$ if there is a substitution
$r\colon X \to X \cup K$ such that $f = r(g)$. We write $f \lep g$ in this case.
(Since $g$ is a polynomial, it only depends on a finite number of indeterminates.
Therefore, we only need to specify a finite part of $r$.)
\item Let $(f_n)$ and $(g_n)$ be p-families. $(f_n)$ is a \emph{p-projection} of $(g_n)$ if
there is a p-bounded function $q\colon \Nset \to \Nset$ such that $f_n \lep g_{q(n)}$.
We write $(f_n) \lep (g_n)$.
\end{enumerate}
\end{definition}

Projections are very simple reductions. Therefore, we can also use them to define hardness for
``small'' complexity classes like $\VP$. More powerful are so-called \emph{c-reductions},
which are an analogue of Turing reductions. c-reductions are strictly more
powerful than p-projections~\cite{DBLP:journals/ipl/IkenmeyerM18}.
Let $g$ be a polynomial in $s$ variables.
A $g$-oracle gate is a gate of arity $s$ that on input $t_1,\dots,t_s$ outputs
$g(t_1,\dots,t_s)$. The size of such a gate is $s$.
$C^g(f)$ denotes the minimum size of a circuit with $g$-oracle
gates that computes $f$. If $G$ is a set of polynomials, then $C^G(f)$ is the minimum
size of an arithmetic circuit that can use any $g$-oracle gates for any $g \in G$.

\begin{definition}
Let $(f_n)$ and $(g_n)$ be p-families. $(f_n)$ is a \emph{c-reduction} of $(g_n)$ if
there is a p-bounded function $q\colon \Nset \to \Nset$ such that $C^{G_{q(n)}}(f_n)$
is p-bounded, where $G_{q(n)} = \{ g_i \mid i \le q(n) \}$.
We write $(f_n) \lec (g_n)$.
\end{definition}

\begin{definition}
A p-family $(f_n)$ is in $\VNP$, if there
are p-bounded functions $p$ and $q$ and a sequence $(g_n) \in \VP$ of polynomials
$g_n \in K[X_1,\dots,X_{p(n)},Y_1,\dots,Y_{q(n)}]$ such that
\[
   f_n = \sum_{e \in \{0,1\}^{q(n)}} g_n(X_1,\dots,X_{p(n)},e_1,\dots,e_{q(n)}).
\]
\end{definition}

$\VP$ and $\VNP$ are algebraic analogues of the classes $\DP$ and $\sharpP$
in the Boolean world. The permanent family $(\per_n)$ is complete for $\VNP$
and the problem of computing the permanent of a given $\{0,1\}$-matrix is complete
for $\sharpP$ under p-projections.

\section{Parameterized (Counting) Complexity}

Parameterized counting complexity was introduced by Flum and Grohe~\cite{DBLP:journals/siamcomp/FlumG04}.
We give a short introduction to fixed parameterized counting complexity.
For more information on parameterized complexity, we refer
the reader to~\cite{DBLP:series/txtcs/FlumG06, DBLP:series/txcs/DowneyF13}.

\begin{definition}
A \emph{parameterized counting problem} is a function $F\colon\Sigma^{*}\times \Nset \rightarrow \Nset$.
\end{definition}

The idea is that an input has two components $(x,k)$,
$x \in \Sigma^*$ is the instance and the parameter $k$ measures the ``complexity'' of the input.

\begin{definition}
A parameterized counting problem is \emph{fixed parameter tractable}
if there is an algorithm computing $F(x,k)$ in time
$f(k)\lvert x\rvert^{c}$ for some computable function $f\colon\Nset \rightarrow \Nset$ and some constant $c$.
The class of all fixed parameter tractable counting problems is denoted by $\sharpFPT$.
\end{definition}

A parameterized counting problem is fixed-parameter tractable if the running time is polynomial
in the instance size. The ``combinatorial explosion'' is only in the parameter $k$.
In particular, the exponent of $n$ does not depend on $k$.
The classical  example for a parameterized counting problem in $\sharpFPT$ is the vertex cover
problem: Given a graph $G$ and a natural number $k$, count all vertex covers
of $G$ of size $k$.

Fixed parameter tractable problems represent the ``easy'' problems in parameterized complexity.
An indication that a problem is not fixed parameter tractable is that it is hard for the
class $\sharpW 1$. Reductions that are used to define hardness are \emph{parsimonious fpt-reductions}:
Such a reduction maps an instance $(x,k)$ to an instance $(x',k')$ such
that the value of the two instances is the same, the running time of the reduction
is $f(k) |x|^c$ for some computable function $f$ and a constant $c$, and there is a computable function
$g$ such that $k' \le g(k)$. It is quite easy to see that the composition
of two parsimonious fpt-reductions is again a parsimonious fpt-reduction
and that $\sharpFPT$ is closed under parsimonious fpt-reductions.

We now define weft $t$ formulas inductively.\footnote{The term ``weft'' originates from textile fabrication
and has been used in Boolean parameterized complexity from its very beginning.} 
\begin{definition}
A \emph{weft} $0$ formula is a layered Boolean formula and the gates have fan-in two (over the basis $\wedge$, $\vee$, and $\neg$).
A weft $t$ formula is a layered Boolean formula where the gates have fan-in two, except one layer of
gates that has unbounded fan-in. This formula has as inputs weft $t-1$ formulas.
\end{definition}

Weft $t$ formulas have $t$ layers of unbounded fan-in gates, and all
other gate have bounded fan-in. Weft $t$ formulas are the defining machine model
of the $\sharpW t$ classes:

\begin{definition}
The class $\sharpW t$ are all parameterized counting problems that are
reducible by parsimonious fpt-reductions to the following problem:
Given a weft t formula $C$ of constant depth and a parameter $k$, count all satisfying assignments
of $C$ that have exactly $k$ $1$s.
\end{definition}

A classical example of a counting problem,
that is $\sharpW 1$-complete,
is counting cliques of size $k$ in a graph.
Clique is used as a major complete problem for $\sharpW 1$ by Flum and Grohe~\cite{DBLP:journals/siamcomp/FlumG04}.
It is known that $\DP = \sharpP$ implies $\sharpFPT = \sharpW 1$.
Curticapean~\cite{DBLP:conf/icalp/Curticapean13} proves
that counting $k$-matchings, the parameterized analogue to the permanent,
is $\sharpW 1$-hard (under fpt Turing reductions).

\section{Parameterized Valiant's Classes}

We now define fixed-parameter variants of Valiant's classes. Our families of polynomials will now
have two indices. They will be of the form $(p_{n,k})$. Here, $n$ is the number of indeterminates and $k$ is
the parameter.

\begin{definition}
A parameterized p-family is a family $(p_{n,k})$ of polynomials such that
\begin{enumerate}
\item $p_{n,k} \in K[X_1,\dots,X_{q(n)}]$ for some p-bounded function $q$, and
\item the degree of $p_{n,k}$ is p-bounded (as a function of $n+k$).
\end{enumerate}
\end{definition}

The most natural parameterization is by the degree: Let $(p_n)$ be any p-family
then we get a parameterized family $(p_{n,k})$ by setting $p_{n,k} = \hp k (p_n)$.
Here $\hp k(f)$ denotes the homogeneous part of degree $k$ of some polynomial $f$.\footnote{I.e., the sum of all monomials of degree $k$ with their coefficients.}
Since $\deg(p_n)$ is polynomially bounded, $p_{n,k}$ is zero when $k$ is large enough.
(This will usually be the case for any parameterization.) More generally, we will also
allow that $p_{n,k} = \hp {t(k)}(p_n)$ for some function $t$ that solely depends on $k$.

Recall that a vertex cover $C$ of a graph $G = (V,E)$ is a subset of $V$ such
that for every edge $e \in E$ at least one endpoint is in $C$.

\begin{example}
Let $\cG = (G_n)$ be a family of graphs such that $G_n$ has $n$ nodes.
We will assume that the nodes of $G_n$ are $\{1,\dots,n\}$.
\begin{enumerate}
\item The vertex cover family $(\VC_n^\cG)$ with respect to $\cG$ is defined as
\[
   \VC_n^{\cG} = \sum_{C \subseteq \{1,\dots,n\}} \prod_{i \in C} X_i
\]
where the sum is taken over all vertex covers $C$ of $G_n$.
\item The parameterized vertex cover family $(\VC_{n,k}^\cG)$, with respect to $\cG$, is defined as
\[
   \VC_{n,k}^{\cG} = \sum_{{\substack{C \subseteq \{1,\dots,n\}\\ {\lvert C\rvert = k}}}} \prod_{i \in C} X_i
\]
\end{enumerate}
where we now sum over all vertex covers of size $k$ of $G_n$. This is a homogeneous polynomial of degree $k$.
(We will call both families $\VC^\cG$. There is no danger of confusion, since
we mainly deal with the parameterized family.)
\end{example}

Every node has a label $X_i$ and for every vertex cover we enumerate (or more precisely, sum up) its weight, which
is the product of the labels of the nodes in it.
Above, every graph family defines a particular vertex cover family.
We can also define a unifying vertex cover family.

\begin{example}
Let $E_{i,j},X_i$, $1 \le i,j\le n$, $i < j$, be variables over some field $K$.
The \emph{parameterized vertex cover polynomial} of size $n$ is defined by
\[
   \VC_{n,k} = \sum_{\substack{C \subseteq \{1,\dots,n\}\\ {\lvert C\rvert = k}}}
               \prod_{\substack{i,j \notin C\\{i < j}}} (1 - E_{i,j}) \prod_{i \in C} X_i.
\]
The parameterized vertex cover family is defined as $(\VC_{n,k})$.
\end{example}

If we set the variables $E_{i,j}$ to values $e_{i,j} \in \{0,1\}$
we get the vertex cover polynomial of the graph given by the adjacency matrix $(e_{i,j})$.
The first product is $0$ if there is an uncovered edge.
More generally, if we take a family of graphs $\cG = (G_n)$ such that $G_n$ has $n$ nodes
and if we plug in the adjacency matrix of $G_n$ into in each $\VC_{n,k}$ then we get the family $(\VC^\cG_{n,k})$.
$(\VC^\cG_{n,k})$ is parameterized by the degree since we have $\VC^\cG_{n,k} = \hp k (\VC_n^\cG$).
$(\VC_{n,k})$, however, is not parameterized by the degree as $\VC_{n,k}$ contains monomials of degree polynomial in $n$ (independent of $k$).

Recall that a clique $C$ of a graph is a subset of the vertices such that for every pair
of nodes in $C$ there is an edge between them.

\begin{example}
\begin{enumerate}
\item Let $E_{i,j},X_i$, $1 \le i,j\le n$, $i < j$, be variables over some field $K$. The \emph{clique
polynomial} of size $n$ is defined by
\[
   \Cl_n = \sum_{C \subseteq \{1,\dots,n\}}
             \prod_{\substack{i,j \in C\\ {i < j}}} E_{i,j} \prod_{i \in C} X_i.
\]
The clique family is defined as $(\Cl_n)$.
\item The parameterized clique family $(\Cl_{n,k})$ is defined by
\[
   \Cl_{n,k} = \sum_{\substack{C \subseteq \{1,\dots,n\}\\ {\lvert  C\rvert = k}}}
             \prod_{\substack{i,j \in C\\ {i < j}}} E_{i,j} \prod_{i \in C} X_i.
\]
\end{enumerate}
(Again, we will call both families $\Cl$.)
\end{example}

If we set the variables $E_{i,j}$ to values $e_{i,j} \in \{0,1\}$,
we get the clique polynomial of the graph given by the adjacency matrix $(e_{i,j})$,
since the first product checks whether $C$ is a clique.
For each clique, we enumerate a monomial $\prod_{i \in C} X_i$. $X_i$ is the label
of the node $i$. $\Cl$ is a polynomial defined on edges and nodes.
This seems to be necessary, since the polynomial
$\sum_{C \subseteq \{1,\dots,n\}} \prod_{i \in C} X_i = (1 + X_1) \cdots (1 + X_n)$,
which is the ``node-only'' version of clique polynomial of the complete graph,
is easy to compute. Therefore, we cannot expect that the ``node-only'' version of the clique family
is hard for some class.

Notice, that the parameterized clique family $(\Cl_{n,k})$ has variables standing in for vertices. These vertices seem to be necessary, as in the counting world, counting the number of $k$ cliques and counting the number of $k$-independent sets are tightly related. Namely, the number of cliques is the number of independent sets on the complement graph. We want to keep this relationship as the problem is an important member of $\#\W{1}$ and hence we incorporate the vertices.

$(\Cl_{n,k})$ is parameterized by the degree,
since $\Cl_{n,k} = \hp {\binom k2 + k} (\Cl_n)$.
Here is another example, beside the general vertex cover family,  of a family
that is parameterized by a different parameter:

\begin{example}
Let $\cG = (G_{n,k})$ be a family of bipartite graphs such that $G_{n,k}$ has $n$ nodes on both sides
and genus $k$, $k \le \lceil (n-2)^2 / 4 \rceil$.\footnote{This is the genus of the $K_{n,n}$~\cite{ringel}.}
Let $A_{n,k}$ be the $n \times n$-matrix that has a variable
$X_{i,j}$ in position $(i,j)$ if there is an edge between $i$ and $j$ in $G_{n,k}$
and a $0$ otherwise.
The \emph{$\cG$-parameterized permanent family} $\per^\cG = (\per_{n,k}^\cG)$ is defined as
$\per_{n,k}^\cG = \per (A_{i,j})$.
\end{example}

There is another natural way to parameterize the permanent:

\begin{example}
Given a $k \times n$-matrix $X = (X_{i,j})$
with variables as entries, the rectangular permanent is defined as
\[
   \rper_{n,k} (X) = \sum_{\substack{f\colon \{1,\dots, k\} \to \{1, \dots,n\}\\ {\text{$f$ is injective}}}} \prod_{i = 1}^k X_{i,f(i)}.
\]
When $k = n$ then this is the usual permanent. The \emph{rectangular permanent family} is defined as $\rper = (\rper_{n,k})$.
\end{example}
We will give some more parameterizations of the permanent in \cref{sec:permanent} where we also prove some hardness results.

We now define fixed parameter variants of Valiant's classes.

\begin{definition}
\begin{enumerate}
\item A parameterized p-family $(p_{n,k})$ is in the class $\VFPT$ if
$C(p_{n,k})$ is bounded by $f(k)p(n)$ for some p-bounded function $p$ and
some arbitrary function $f\colon \Nset \to \Nset$.\footnote{$f$ need not be computable,
since Valiant's model is nonuniform.}
\item The subclass of $\VFPT$ of all parameterized p-families that are parameterized
by the degree is denoted by $\VFPTd$.
\end{enumerate}
\end{definition}

We will also say above that $C(p_{n,k})$ is \emph{fpt-bounded}.
We will see in one of the next sections that the vertex cover family and the
bounded genus permanent are in $\VFPT$. We will say that a family of circuits $(C_n,k)$ has \emph{fpt size} if the size is bounded by $f(k)p(n)$ for some function $f\colon\Nset \to \Nset$ and p-bounded function $p$.

\begin{definition}
A parameterized p-family $f = (f_{n,k})$ is an \emph{fpt-projection} of another parameterized p-family $g = (g_{n,k})$
if there are functions $r,s,t\colon \Nset \to \Nset$ such that $r$ is p-bounded and
$f_{n,k}$ is a projection of $g_{r(n)s(k),k'}$ for some $k' \le t(k)$.\footnote{$k'$
might depend on $n$, but its size is bounded by a function in $k$.
There are examples in the Boolean world, where this dependence on $n$ is used.}
We write
$f \lefp g$.
\end{definition}

\begin{lemma}\label{lem:p:closed}
If $f \in \VFPT$ (or $\VFPTd$) and $g \lefp f$, then $g \in \VFPT$ (or $\VFPTd$, respectively).
\end{lemma}
\begin{proof}
  The proof essentially follows from the fact that fpt-bounded functions
  are closed under composition.
  Let the functions $r,s,t$ be as in the definition above.
  We know that for all $n,k$, $f_{n,k}(X_1,\dots,X_{p(n)}) = g_{r(n)s(k),k'}(Y_1,\dots, Y_{q(r(n)s(k))})$
  where each $Y_i \in K \cup \{X_1,\dots,X_{p(n)}\}$.
  Since $(g_{n,k}) \in \VFPT$, $C(g_{n,k}) \le u(k) m(n)$ for some p-bounded $m$ and an arbitrary
  function $u$.
  Then $C(f_{n,k}) \le u(k') m(r(n)s(k)) \le u(t(k)) m(r(n)s(k))$.
  Thus $C(f_{n,k})$ is fpt-bounded. (Note that $m(r(n)s(k))$ can be written as a product
  of a polynomial in $n$ and a function depending only on $k$, since $m$ is p-bounded.)
  Notice, that in the case of $\VFPTd$ the same argument holds as the projection can only lower
  the degree.
  \end{proof}

\begin{lemma}\label{lem:p:trans}
$\lefp$ is transitive.
\end{lemma}
\begin{proof}
  This is essentially the same proof as for p-projections, but instead
  of using that \mbox{p-bounded} functions are closed under composition, we use that fpt-bounded
  functions are closed under composition.
  Let $f \lefp g \lefp h$. For all $n,k$
  \[
     f_{n,k}(X_1, \dots, X_{p_f(n)}) = g_{r(n)s(k),k'}(Y_1,\ldots, Y_{p_g(r(n)s(k))})
  \]
  and
  \[
      g_{n,k}(X_1,\dots,X_{p_g(n)}) = h_{r'(n)s'(k), k''}(Y'_1,\dots,Y'_{p_h(r'(n)s'(k))})
  \]
  with p-bounded functions $p_f$, $p_g$, $p_h$, $r$, and $r'$, $k'\leq t(k)$ and $k''\leq t'(k)$
  for some arbitrary functions  $t$ and $t'$, and the $Y_i$ and  $Y'_i$ being constants or variables.
  Now the transitivity is obvious, as we can write $f_{n,k}$ as a projection of
  $h_{r'(r(n)s(k))s'(k'),k''}$ by substituting the $Y_i$ into the $Y'_i$,
  where $k''\leq t'(t(k))$.
  \end{proof}

One can define a notion of completeness. In the case of fpt-projections,
the degree of the polynomial is the only meaningful parameter to consider:
The permanent family on bounded genus graphs $\per^{\cG}$
is in $\VFPT$ and so is (a variant of) the vertex cover family $\VC$. However, every polynomial in the permanent family
has degree equal to the number of nodes in the graph (independent of the genus) whereas
the degree of the vertex cover polynomial depends on the degree. If a polynomial $p$ is a projection of $q$,
then $\deg p \le \deg q$. Therefore, $\per^{\cG}$ cannot be an fpt-projection of $\VC$.
Now we can call a parameterized family $f$ \emph{$\VFPTd$-complete} (under fpt-projections),
if it is in $\VFPTd$ and for all $g \in \VFPTd$, $g \lefp f$.

For other parameters, we need a stronger notion of reduction.
There are the so-called c-reductions, see~\cite{Buergisser:00}, which are the analogue of
Turing reductions in Valiant's world. This is the strongest kind of reduction one could use.
However, the p-projections in Valiant's world
seem to be weaker than parsimonious polynomial-time reductions in the
Boolean world. Therefore, we propose an intermediate concept,
which models parsimonious reductions in the algebraic world. In parsimonious reductions, the input instance
is transformed by a polynomial time or fpt computable reduction, then the function we reduce to is evaluated,
and the result that we get shall be the value of our given function evaluated at the original instance.

In the algebraic world, this can be modeled as follows:
We call a p-family $f = (f_n)$ with $f_n \in K[X_1,\dots,X_{p(n)}]$ a \emph{p-substitution}
of a p-family $g = (g_n)$ with $g_n \in K[X_1,\dots,X_{q(n)}]$
if there is a p-bounded function $r$, and for all $n$, there
are $h_1,\dots,h_{q(r(n))}$
such that $f_n = g_{r(n)}(h_1,\dots,h_{q(r(n))})$ and $\deg(h_i)$ as well as $C(h_i)$ is p-bounded for all $i$.
We write $f \les g$.
Compared to a projection, we are now allowed to substitute polynomials of
p-bounded complexity. We have that $\les$ is transitive and that $p \les q$ and
$q \in \VP$ implies $p \in \VP$.

The parameterized analogue is defined as follows.

\begin{definition}
A parameterized p-family $f = (f_{n,k})$ with $f_{n,k} \in K[X_1,\dots,X_{p(n)}]$
is an \emph{fpt-substitution}
of another parameterized p-family $g = (g_{n,k})$ with $g_{n,k} \in K[X_1,\dots,X_{q(n)}]$
if there are functions $r,s,t\colon \Nset \to \Nset$ such that for all $n,k$, $r$ is p-bounded
and there exists polynomials $h_1$,$\dots$,$h_{q(r(n)s(k))}$ $\in K[X_1,\dots,X_{p(n)}]$
such that
\[
   f_{n,k} = g_{r(n)s(k),k'}(h_1,\dots,h_{q(r(n)s(k))})
\]
for some $k' \le t(k)$ and $\deg(h_i)$ as well as $C(h_i)$ is fpt-bounded (with respect to $n$ and $k$)
for all $i$.
We write $f \lefs g$.
\end{definition}

\begin{lemma}\label{lem:s:closed}
If $f \in \VFPT$ and $g \lefs f$, then $g \in \VFPT$. \qed
\end{lemma}

\begin{lemma}\label{lem:s:trans}
$\lefs$ is transitive. \qed
\end{lemma}

To define an algebraic analogue of $\sharpW t$,
we study unbounded fan-in arithmetic circuits.
These circuits have multiplication and addition gates of arbitrary fan-in.
A gate with fan-in $2$ will be called a gate of bounded fan-in, any other
gate is a gate of unbounded fan-in. (Instead of $2$, we can fix any other bound $b$,
and we will do this in the  next sections to make the presentation more convenient.)

\begin{definition}
Let $C$ be an arithmetic circuit. The \emph{weft} of $C$ is the maximum
number of unbounded fan-in gates on any path from a leaf to the root.
\end{definition}

For $s,k \in \Nset$, $\ones s k$ denotes the set of all $\{0,1\}$-vectors
of length $s$ having exactly $k$ $1$s.

\begin{definition}
\begin{enumerate}
\item A parameterized p-family $(f_{n,k})$ is in $\VW t$, if there
is a p-family $(g_{n})$ of polynomials
$g_{n} \in K[X_1,\dots,X_{p(n)},Y_1,\dots,Y_{q(n)}]$ with p-bounded $p$ and $q$
such that $g_n$ is computed by a constant depth unbounded fan-in circuit of weft $\le t$ and polynomial size
and
\begin{equation} \label{eq:def:vw}
   (f_{n,k}) \lefs \Big( \sum_{e \in \ones{q(n)}{k}} g_{n}(X_1,\dots,X_{p(n)},e_1,\dots,e_{q(n)}) \, \Big).
\end{equation}
\item $\VWd t$ is the subset of all families in $\VW t$, that have the degree as the parameter.
\end{enumerate}
\end{definition}

In essence, we emulate the Boolean $\#\W t$ definition. Instead of Boolean circuits of weft $t$ we take an arithmetic circuit and instead of counting the number of assignments, we sum over all assignments. In addition, we only count the assignments that have weight $k$ by adjusting the vectors we sum over, namely to $\{0,1\}$-vectors with exactly $k$ ones. While in the Boolean setting the closure is taken with respect to parsimonious fpt-reductions, in the arithmetic setting, we take fpt-substitutions. Hence, our definition seems to be the most appropriate analogue.

The clique family is in $\VW 1$, since we can write it as
\[
  \Cl_{n,k} = \sum_{v \in \ones nk }
             \prod_{\substack{i,j = 1\\ {i < j}}}^n (E_{i,j}v_i v_j + 1 - v_i v_j)  \prod_{i = 1}^n (X_i v_i + 1 - v_i).
\]
This formula has weft 1, since there are two unbounded product gates and none is a predecessor of the other.
We replace the product over all $C$ by a product over all vertices and use the $v$-vectors to switch
variables on and off.

Like in the Boolean case,
we will show that the parameterized clique family is complete for the class $\VW 1$ (albeit
for a stronger notion of reductions, namely fpt-c-reductions).
It turns out that this proof is far more complicated than in the Boolean setting, since our circuits
can compute arbitrary polynomials and not only Boolean values. Furthermore, multiplication and addition
cannot be reduced to each other since there is no analog of de~Morgan's law.

\begin{definition}
Let $f = (f_{n,k})$ and $g = (g_{n,k})$ be parameterized p-families. $f$ \emph{fpt-c-reduces} to $g$ if
there is a p-bounded function $q\colon \Nset \to \Nset$ and functions $s,t\colon \Nset \to \Nset$
such that $C^{G_{q(n)s(k),t(k)}}(f_{n,k})$
is fpt-bounded, where $G_{q(n)s(k),t(k)} = \{ g_{i,j} \mid i \le q(n)s(k), \enspace j \le t(k) \}$.
We write $f \lefc g$.
\end{definition}

The following two lemmas are proved like for $\lec$ and $\VP$.
We replace oracle gates by circuits and use the fact that fpt-bounded
functions are closed under composition.

\begin{lemma}
If $f \in \VFPT$ and $g \lefc f$, then $g \in \VFPT$. \qed
\end{lemma}

\begin{lemma}
$\lefc$ is transitive. \qed
\end{lemma}

So we have two different notions to define $\sharpW t$-hardness.
Presumably, they are different, see~\cite{DBLP:journals/ipl/IkenmeyerM18}.

\section{$\VFPT$}

\begin{theorem}\label{thm:fptv:vertexcoverfpt}
For every family of graphs $\cG = (G_n)$, where $G_n$ has $n$ nodes,
$\VC_{n,k}^{\cG}$ is in $\VFPTd$.
\end{theorem}
\begin{proof}
  Given $G=(V,E)$ with $V = \{1,\dots,n\}$ and a parameter $k$,
  we work in the same way as the classical FPT algorithm. We take an arbitrary order of
  edges. For every edge, we branch on which vertex we add to our cover and recurse. We define for
  an edge set $E_{\setminus \{v\}}$ to denote the set
  $\{(u,u') \in E\mid u\neq v\text{ and } u'\neq v\}$.
  Let $P(G,k)$ denote the $k$-vertex cover polynomial of $G$. We write a recursive formula
  for $P(G,k)$:
  \[
          P(G,k) = \begin{cases}
              1 & \text{if $k=0$ and $E=\emptyset$},\\
              0 & \text{if $k=0$, $E\neq \emptyset$},\\
              \sum_{i_1 < \dots < i_k \in V} X_{i_1} \cdots X_{i_k}& \text{if $k\neq 0$, $E=\emptyset$},\\
          \begin{aligned}
              &X_{u}\cdot P((V\setminus\{u\}, E_{\setminus\{u\}}),k-1)\\
              &\qquad +X_{v}P((V\setminus\{v\}, E_{\setminus\{v\}}),k-1) \\
              &\qquad - X_u X_v P((V\setminus\{u,v\}, E_{\setminus\{u,v\}}),k-2)
          \end{aligned}  &\text{otherwise.}
      \end{cases}
      \]
  The first line of the above recursive definition means
  that the vertex cover polynomial of the empty graph is $1$.
  If there are edges left after $k$ recursive steps, we remove this computation subtree by
  multiplying it with zero (second line).
  If the there are no edges, then the vertex cover polynomial is the $k$th elementary
  symmetric polynomial on the vertex set $V$ (third line).
  Finally, we branch on a chosen edge $e = \{u,v\}$ in the fourth line:
  At least one of the two nodes $u$ and $v$ has to cover the chosen edge.
  We subtract the third term, since the case when both $u$ and $v$ cover the edge is counted
  twice by the first two terms.
  It is easy to see that we always get the same polynomial independent
  of the ordering of the edges.
  The number of operations and hence the size of the circuit is bounded by $T(n,k) \le 3T(n,k-1) + O(1)$
  plus $\poly(n)$ for computing the elementary symmetric polynomial in the third case.\footnote{The polynomial is just the homogeneous components of $t$ of degree $n-k$ in $(t+X_1)\cdots (t+X_n)$.}
  Thus $T(n,k)$ is fpt-bounded.
  \end{proof}

\begin{remark}
It is unlikely that the general family $\VC_{n,k}$ is in $\VFPT$.
Take any graph $G = (V,E)$ on $n$ nodes and $m$ edges and compute $\VC_{n,k}$ on this graph. Now, for $i < j$, we set
\[
  E_{i,j} = \begin{cases}
              1 - S & \text{if $\{i,j\} \in E$}, \\
              0     & \text{otherwise},
            \end{cases}
\]
and $X_i = T$ for all $i$. Then we get a bivariate polynomial. This polynomial
contains a monomial $S^i T^j$ iff there is a vertex cover of size $j$ in $G$
not covering $i$ edges, or, equivalently, covering $m - i$ edges. Note that since
the polynomial is now bivariate, we can easily compute its coefficients using
interpolation. While the (Boolean decision version of) vertex cover is in $\mathsf{FPT}$,
it turns out~\cite{DBLP:journals/mst/GuoNW07} that the more general question whether there is a set of nodes of size $k$
covering at least $t$ edges is $\mathsf{W}[1]$-hard (with parameter $k$).
Therefore, it seems to be unlikely that $\VC_{n,k}$ has circuits of fpt size.

Mahajan and Saurabh~\cite{DBLP:conf/csr/MahajanS16}
define another variant of the vertex cover polynomial. We multiply each cover
by a product over the uncovered edges. They multiply by a product over the covered edges.
Both polynomials are essentially equivalently, one can turn one into the other by dividing
through the product over all edges, doing a variables transform, and removing divisions.
\end{remark}

The \emph{sun graph} $S_{n,k} = (V,E)$ on $n$ nodes is defined as follows: The first
$2k$ nodes form a clique. And every other node is connected to the nodes $1,\dots,2k$,
but to no other nodes, that is, the nodes $2k + 1,\dots,n$ form an independent set.
Every graph $G$ with $n$ nodes that contains a vertex cover of size $k$ is a subgraph of $S_{n,k}$.
To see this, we compute a maximum matching $M$ in $G$. This matching $M$ has size at most $k$, since
at least one endpoint of each edge in $M$ has to be in a vertex cover. There cannot be any edge between two nodes that are not
endpoints of an edge in $M$, since $M$ is maximum. Therefore, all nodes not in $M$ form an independent set.
We map the nodes in $M$ to the nodes $1,\dots,2k$ in $S_{n,k}$ and all other nodes to the nodes $2k + 1,\dots,n$.

We define $\sVC_{n,k}$ like $\VC_{n,k}$ but on the graph $S_{n,k}$ instead of $K_n$.
The difference to $\VC$ is, that we now have some idea where the vertex
cover is located (like it is in the Boolean case where we can find a potential set for instance by computing a maximum matching).
Therefore, we can obtain:

\begin{theorem}\label{thm:vcsun}
$\sVC \in \VFPT$.
\end{theorem}
\begin{proof}
  Let $C \subseteq \{1,\dots,2k\}$ with $\kappa \coloneqq \lvert C\rvert \le k$. The contribution to $\sVC_{n,k}$ of $C$ is
  \[
     \prod_{\substack{i,j \in \{1,\dots,2k\} \setminus C\\ {i < j}}} (1 - E_{i,j}) \prod_{i \in C} X_i .
  \]
  To the set $C$, we can add $k - \kappa$ nodes from $\{2k+1,\dots,n\}$.
  Consider the product
  \[
     \prod_{j = 2k + 1}^n \left(X_j + \prod_{i \notin C} (1 - E_{i,j}) \right).
  \]
  We go over all nodes in $\{2k+1,\dots,n\}$. For each node $j$, we either add it to the cover or
  the edges incident with $j$ are uncovered unless they are covered by the other node $i \in \{1,\dots,2k\}$.
  The homogeneous component of degree $k - \kappa$ in the $X$-variables is the contribution
  to $\sVC_{n,k}$ by the outer nodes. Altogether, we have
  \[
     \sVC_{n,k} = \sum_{\substack{C \subseteq \{1,\dots,2k\}\\  {\lvert C\rvert \le k}}}
       \left(\prod_{\substack{i,j \in \{1,\dots,2k\} \setminus C\\ {i < j}}} (1 - E_{i,j}) \prod_{i \in C} X_i \right)
       \cdot H_{k - \lvert C\rvert,X} \left( \prod_{j = 2k + 1}^n \left(X_j + \prod_{i \notin C} (1 - E_{i,j}) \right)  \right).
  \]
  Here $H_{k - \lvert C\rvert,X}$ denotes the homogeneous components of degree $k - \lvert C\rvert$ in the $X$-variables.
  The number of summands is bounded by a function $k$ and each summand can be computed by a
  polynomial sized circuit. Thus $\sVC_{n,k} \in \VFPT$.
  \end{proof}

Both parameterized permanent families turn out to be fixed parameter tractable.

\begin{theorem}\label{thm:genusper}
For every family of bipartite graphs $\cG = (G_{n,k})$ such that $G_{n,k}$ has
$n$ nodes and genus $k$, $\per^\cG$ is in $\VFPT$.
\end{theorem}
\begin{proof}
  Galluccio and Loebl~\cite{DBLP:journals/combinatorics/GalluccioL99} prove that the permanent
  of a bipartite graph with genus $k$ can we written as a linear combination of $4^k$ determinants
  of $n \times n$-matrices whose entries are linear forms. Since the determinant has
  polynomial size circuits, the claim follows.
  \end{proof}

\begin{theorem}\label{thm:rper}
$\rper \in \VFPT$.
\end{theorem}
\begin{proof}
  Vassilevska and Williams~\cite{DBLP:conf/stoc/VassilevskaW09} construct an arithmetic circuit
  for $\per_{n,k}$ of size $O(k 2^k n^3)$.
\end{proof}

Kernelization is an important concept in parameterized complexity. In the algebraic
setting, $\VFPT$ can also be characterized by kernels of size $f(k)$.

\subsection{Kernelization}\label{sec:fptv:kernelization}

In the Boolean world, a \emph{kernelization} is a special kind of reduction.
Given a parameterized decision problem $L: \Sigma^* \times \Nset \to \{0,1\}$,
a kernelization is a polynomial time computable
mapping (reduction) $r\colon \Sigma^* \times \Nset \to \Sigma^* \times \Nset$
such that
\begin{itemize}
\item $(x,k) \in L$ $\iff$ $r(x,k) \eqqcolon (x',k') \in L$,
\item $\lvert x'\rvert$ is bounded by a computable function in $k$ and $k'$ is bounded by a
computable function in $k$.
\end{itemize}
So in polynomial time, $x$ is reduced to an instance $x'$ the size of which
only depends on $k$. $x'$ can be solved using brute-force and $x'$ is often called
the problem kernel.
It is know that a decision problem is in $\FPT$ if and only if it is decidable and admits a kernelization.

Kernelizations for counting problems or enumeration problems are not well
understood. Consider the vertex cover problem: A graph with $n$ nodes can have
a number of vertex covers that depends on $n$, but a kernel can only have a number of vertex covers
that depends solely on $k$. Therefore, parsimonious reductions cannot work.
See~\cite{DBLP:conf/tamc/Thurley07} for some work on counting kernels
and~\cite{DBLP:conf/mfcs/CreignouMMSV13} for a definition of enumeration kernels.

\begin{definition}\label{def:fptv:vkernel}
Let $(f_{n,k})$ be a parameterized $p$-family such that
$f_{n,k} \in K[X_1,\dots,X_{p(n)}]$. A $\emph{kernelization}$
consists of functions $r\colon \Nset^2 \to \Nset$ and $s, t\colon \Nset \to \Nset$ such that for all $n$ and $k$,
$r(n,k) \le s(k)$ and
there are circuits $C_{n,k,i}$ of size polynomial in $n$ computing polynomials $q_{n,k,i}$, $1 \le i \le p(r(n,k))$,
and there is a $k' \le s(k)$ such that
\begin{equation} \label{eq:fptv:vkernel}
   f_{n,k}(X_1,\dots,X_{p(n)}) = f_{r(n,k),k'}(q_{n,k,1},\dots,q_{n,k,p(r(n,k))}).
\end{equation}
\end{definition}

Essentially, this definition means that every $f_{n,k}$
is a p-substitution of some $f_{n',k'}$ and $n'$ and $k'$ are both bounded by a function in $k$.

We can also prove a general theorem similar to the Boolean setting:

\begin{theorem}\label{thm:fptv:veqkernelfpt}
Let $f$ be a parameterized p-family.
\begin{enumerate}
\item If $f$ admits a kernelization, then $f \in \VFPT$.
\item If $f \in \VFPT$  and there are $n_0$ and $k_0$ such that $f_{n_0,k_0}$
is linear in some variable $X_i$, then $f$ admits a kernelization.
\end{enumerate}
\end{theorem}

\begin{proof}
For the first item, assume we are given $n$ and $k$ and let
$r$, $s$, and $t$ and $q_{n,k,i}$, $1 \le i \le p(r(n,k))$ as well as $k'$ as in \cref{def:fptv:vkernel}.
The number of variables and the degree of $f_{r(n,k),k'}$ only depends on $k$. Therefore, its circuit
complexity is also bounded by a function of $k$, since we can simply take the trivial circuit that computes all monomials
and sums them up. Since the $q_{n,k,i}$ have polynomial size circuits, $C(f_{n,k})$
is fpt-bounded by (\ref{eq:fptv:vkernel}).

For the second item, let $C(f_{n,k})$ be bounded by $u(k) m(n)$ for some p-bounded function $m$.
We can assume that $m(n) \ge n$ for all $n$.
By substituting appropriate
constants for all other variables, we get that $X_i$ is a projection of $f_{n_0,k_0}$.
Now given $n$ and $k$, if $m(n) \le u(k)$, then $r(n,k) = n$ and $k' = k$. Note that now $r(n,k) = n \le u(k)$.
If $u(k) < m(n)$, then we can compute $f_{n,k}$ by a circuit of size $u(k) m(n) \le m^2(n)$, which is
polynomial. Since $f_{n,k}$ trivially is a substitution of $X_1$, we can set $r(n,k) = n_0$ and $k' = k_0$,
which are both constants.
\end{proof}

While the proof feels somewhat unsatisfactory, we remark that the proof in the Boolean case works along the same
lines. We can remove the degree condition in the second item of \cref{thm:fptv:veqkernelfpt} by
relaxing the notion of kernel and allowing some postprocessing. In the general case, $X_i$ will have
degree $d$ and we only get that $X_i^d$ is a projection of $f_{n_0,k_0}$ and therefore,
we compute $f_{n,k}^d$ in the second case of the case distinction. It is however known that when
a circuit of size $s$ computes some power $f^d$ of a polynomial $f$, then there is a circuit
of size polynomial in $s$ that computes $f$ by using Newton iteration~\cite{DBLP:journals/jacm/BrentK78}
(see~\cite{Ulrich} for an explicit construction).

\section{The $\mathsf{VW}$-hierarchy}

We start with proving some basic facts about the $\VW t$ classes,
in analogy to the Boolean world.

\begin{lemma}\label{lem:hier:1}
$\VFPT = \VW 0$ and $\VFPTd = \VWd 0$.
\end{lemma}
\begin{proof}
  The proof is obvious, since $\VW 0$ and $\VWd 0$ are defined as the closure under
  fpt-substitutions, so we can compute problems in $\VFPT$ simply by using the reduction.
\end{proof}
The following lemma is obvious.

\begin{lemma}
For every $t$, $\VW t \subseteq \VW {t+1}$ and $\VWd t \subseteq \VWd {t+1}$.
\end{lemma}

We call a parameterized p-family $f$ \emph{$\VW t$-hard} (under fpt-substitutions), if for all
$g \in \VW t$, $g \lefs f$. $f$ is called \emph{$\VW t$-complete} (under fpt-substitutions) if in addition,
$f \in \VW t$. If the same way, we can also define hardness and completeness under fpt-c-reductions.

For the classes $\VWd t$, it is reasonable to study hardness and completeness under fpt-projections.
We call a parameterized p-family $f$ \emph{$\VWd t$-hard} (under fpt-projections), if for all
$g \in \VW t$, $g \lefp f$. $f$ is called \emph{$\VWd t$-complete} (under fpt-projections) if in addition,
$f \in \VW t$.

\begin{lemma}\label{lem:hier:2}
If $f$ is $\VW {t+1}$-complete under fpt-substitutions and $f \in \VW t$, then $\VW t = \VW {t+1}$.
In the same way,
if $f$ is $\VWd {t+1}$-complete under fpt-substitutions or fpt-projections and $f \in \VWd t$, then $\VWd t = \VWd {t+1}$.
\end{lemma}
\begin{proof}
  Let $g \in \VW {t+1}$ be arbitrary. By the completeness of $f$, $g \lefs f$. Since
  $f \in \VW t$ and $\lefs$ is transitive, $g \in \VW t$. The same proof works for
  $\VWd t$.
\end{proof}

It is open in the Boolean case whether $\W t = \W {t+1}$ or $\sharpW t = \sharpW {t+1}$
implies a collapse of the corresponding hierarchy. Maybe the algebraic setting can provide
more insights.

\begin{theorem}\label{thm:hier:3}
If $\VFPT \not= \VW 1$ then $\VP \not= \VNP$.
\end{theorem}
\begin{proof}
  Assume that $\VP = \VNP$. Then the clique family $(\Cl_n)$ has polynomial sized
  circuits, since it is in $\VNP$. It follows that also the parameterized clique family
  $(\Cl_{n,k})$ has polynomial sized circuits. To see this, we can replace the variables $X_i$ by
  $T X_i$ for some new variable $T$. View the resulting polynomial as a univariate polynomial in $T$
  with coefficients being polynomials in the original variables. The coefficient of $T^k$ is
  $\Cl_{n,k}$. Since interpolation can be done with polynomial circuits,
  $(\Cl_{n,k})$ has polynomial sized circuits. We will prove that $\Cl$ is $\VW 1$ complete
  under fpt-c-reductions (\cref{thm:cl:w1}). Therefore, every family in $\VW 1$ has circuits
  of fpt size. This proves the theorem.
\end{proof}

If one takes the defining problems for $\VW t$ (sums over $\{0,1\}$ vectors with
$k$ $1$s of weft $t$ circuits) instead of clique, one can prove the same theorem
for arbitrary classes $\VW t$ in place of $\VW 1$. The proof only get technically a little
more complicated.

\section{Hardness of Clique}\label{sec:hardnessclique}

Our main technical result is the $\VW 1$-hardness of $\Cl$. The proof is
technically much more intricate than in the Boolean setting.
We will give a short outline.

\begin{itemize}
\item First, we prove as a technical tool that two bounded exponential sums over a weft $t$ circuit
can be expressed by one exponential sum over a (different) weft $t$ circuit. In the case of $\VNP$, a similar proof is easy:
Instead of summing over bit vectors of length $p$ and then of length $q$, we can sum over bit vectors
of length $p + q$ instead. If the number of ones is however bounded, this does not work easily anymore.
It turns out that for the most interesting class $\VW 1$ of the $\mathsf{VW}$-hierarchy, the construction
is astonishingly complicated. See \cref{app:comp}.
\item Next, we prove a normal form for weft $1$ circuits. Every weft $1$ circuit can be replaced by an equivalent
weft $1$ circuit that has five layers: The first layer is a bounded summation gate, the second layer consist of bounded
multiplication gates, the third layer is the only layer of unbounded gates, the fourth layer again
consists of bounded addition gates and the fifth layer of bounded multiplication gates. See \cref{app:weft}.
\item Then we introduce \emph{Boolean-arithmetic formulas}: A Boolean-arithmetic formula is a formula of the form
\[
    B(X_1,\dots,X_n) \cdot \prod_{i = 1}^n (R_i X_i + 1 - X_i)
\]
where $B$ is an arithmetization of some Boolean formula and the $R_i$ some polynomial
or even rational function (over a different set of variables). For each satisfying $\{0,1\}$-assignment $e$
to $B$, that is, $B(e) = 1$, the right hand side produces one product  and the $e_i$ (assigned to the $X_i$ variables) switch the factors $R_i$ on or off.
For a polynomial $f$, the monomials of support size $k$ are all monomials that depend on exactly $k$ variables.
The sum of all these monomials is denoted by $\spc k(f)$. A central result for the hardness proof
is that when $f$ is computed by a circuit of weft $1$, then
we can write $\spc k(f)$ as a bounded sum over a weft $1$ Boolean arithmetic expression, that is,
$
 \spc k(f) = \sum_{e \in \ones {p(n,k)}{q(k)}}
               \cB (e) \cdot \prod_{i = 1}^{p(n,k)} (R_i e_i + 1 - e_i)$.
See \cref{app:bool}.
\item Finally, we prove in \cref{app:cliq}, that $\Cl$ is $\VW 1$-complete under
fpt-c-reductions (or under fpt-substitutions that allow rational expressions).
Given some bounded sum over a polynomial $g_n(X_1,\dots,X_p,Y_1\dots,Y_q)$ computed by a weft 1 circuit,
we view $g_n$ as a polynomial over the $Y$-variables, the coefficients of which are polynomials in the $X$-variables.
Then $\spc 0(g),\dots,\spc k(g)$ are the parts of $g_n$ that contribute to the sum when summing over all
bit vectors with $k$ ones. We can write this as a bounded sum over a Boolean arithmetic formula.
The concept of Boolean arithmetic formulas allows us to reuse parts of the Boolean hardness proof.
\end{itemize}

Note that once we have the $\VW 1$-hardness of $\Cl$, we get further hardness results for free.
In many counting problems, we want to count subsets of some given set having a certain property.
For instance, in the clique problem, we are given a graph and count node sets of size $k$ that
form a clique. To simplify the presentation, let us now consider such graph problems where we want to
count node sets of size $k$ having a particular property. Whether a given node set has a certain
property is a Boolean function of the adjacency matrix of the graph. We get a generic way to construct a corresponding
family of polynomials by:
\[
   V_{n,k} = \sum_{e \in \ones n k} P_e(E)  X_1^{e_1}\cdots X_n^{e_n}.
\]
The polynomial has edges variables $E_{1,1},\dots,E_{n,n}$ and node variables $X_1,\dots,X_n$.
We sum over all node sets of size $k$ and the product produces a label of the corresponding node set.
The polynomial $P_e(E)$ in the edge variables is an arithmetization of the Boolean expression
that checks whether the node set given by $e$ has indeed the required property. The $P_e$ can all be different,
but in most cases, they can be obtained from each other by permuting the edge variables, like it is for $\VC$ or $\Cl$.

Another example of this kind is the grid tiling problem:
Given a graph $G$ and a parameter $k$, count the subgraphs of $G$ that are isomorphic to a
$k \times k$-grid. The corresponding family of polynomials looks as follows:
\begin{multline*}
   \mathrm{GT}_{n,k} = \sum_{\substack{(a_{i,j}) \in \{1,\dots,n\}^{k \times k}\\ \text{$a_{i,j}$ pairwise distinct, $a_{1,1} < a_{1,n}, a_{n,1}$}}}
              \prod_{1 \le i,j < k} E_{\{a_{i,j}, a_{i+1,j}\}} E_{\{a_{i,j}, a_{i,j+1}\}} \\
              \prod_{1 \le i < k} E_{\{a_{i,k}, a_{i+1,k}\}} E_{\{a_{k,i}, a_{k,i+1}\}}\prod_{1 \le i, j  \le k} X_{a_{i,j}}.
\end{multline*}
The first two products checks whether the selected vertices
indeed form a grid in the given graph. The third product writes down a label of the grid.
The condition $a_{1,1} < a_{1,n}, a_{n,1}$ is due to the fact that the grid has four automorphisms.

\begin{definition}
We call a (Boolean) parsimonious fpt-reduction  between two parameterized node counting
problems $(F,k), (F',k')$  \emph{witness projectable} if all instances $I$
with $n$ nodes are mapped to instances $I'$ with exactly $s(n,k)$ nodes
and for every $n$, there is a mapping $p_{n,k}\colon \{1,\dots,n\} \to \{1,\dots,s(n,k)\}$
such that for every solution $S'$ to $I'$, $S = \{ u \mid p_{n,k}(u) \in S'\}$ is a solution to $I$.
\end{definition}

Witness projectable reductions have the property that we get a solution $S$ to $F$ from a solution $S'$ to $F'$
by a simple projection, that is, one can view $S$ as a subset of $S'$. Furthermore, witness projectable also
implies that the reduction maps two instances of length $n$ to two instances having the same length $s(n,k)$.
In particular, every bit of $I'$ (which is an entry in the adjacency matrix)
is a function of the bits of $I$. We can now prove the simple proposition.

\begin{proposition}\label{prop:wp}
Assume we have a witness projectable reduction from the Boolean $k$-clique problem
to a parameterized node counting problem $(F,k')$. Then the node counting polynomial corresponding to $F$
as defined above,
  \[
    V_{n,k} = \sum_{e\in \ones{n}{k}} P_e(E) X_1^{e_1}\cdots X_n^{e_n}
  \]
is $\VW 1$ hard.
\end{proposition}
\begin{proof}
  We want to solve $(\Cl_{n,k})$ with $(V_{n,k})$.
  Let $(I,k)$ be a Boolean clique instance and $(I',k')$ be its image under the reduction.
  Since fpt-parsimonious witness projectable reductions map instances with $n$ nodes
  to instances with $s(n,k)$ many nodes, the entries of the adjacency matrix
  of $I'$ are Boolean functions in the entries of adjacency matrix of $I$ and these expressions are independent of $I$.
  We arithmetize the corresponding Boolean formulas and substitute the edge
  variables of $V_{s(n,k),k'}$ by these expressions. By the second property
  of witness projectable, we get $\Cl_{n,k}$ by setting some of the $X$-variables to $1$,
  namely, those not in the image of $p_{n,k}$. The hardness now follows from transitivity.
\end{proof}

In particular, the family $(\mathrm{GT}_{n,k})$ is $\VW 1$-hard, since there is a witness projectable
fpt-parsimonious reduction from the Boolean clique problem to the Boolean grid tiling problem.
Instead of a $k \times k$-grid, any graph of sufficiently high tree width would work.

\subsection{Composition of Sums}\label{app:comp}

If $(f_n(X_1,\dots,X_n,Y_1,\dots,Y_{p(n)}, Z_1,\dots,Z_{q(n)}))$ is in $\VP$,
then the family $(g_n)$ given by
\begin{equation}\label{eq:w1:1}
    g_n = \sum_{d \in \{0,1\}^{p(n)}} \sum_{e \in \{0,1\}^{q(n)}}
            f(X_1,\dots,X_n,d_1,\dots,d_{p(n)},e_1,\dots,e_{q(n)})
\end{equation}
is in $\VNP$, since we can write it as
\begin{equation}\label{eq:w1:2}
    g_n = \sum_{e \in \{0,1\}^{p(n) + q(n)}}
           f(X_1,\dots,X_n,e_1,\dots,e_{p(n)},e_{p(n)+1},\dots,e_{p(n) + q(n)}).
\end{equation}

If the family $f_n$ is instead computed by polynomial size constant depth circuits of weft~$t$
and we define the parameterized family $(g_{n,k})$ by
\begin{equation} \label{eq:w1:3}
   g_{n,k} = \sum_{d \in \ones {p(n)} k} \sum_{e \in \ones {q(n)} k}
            f(X_1,\dots,X_n,d_1,\dots,d_{p(n)},e_1,\dots,e_{q(n)}).
\end{equation}
Then we cannot combine these two sums as in (\ref{eq:w1:2}) by summing over
$e \in \ones {p(n) + q(n)}{2k}$, since we then also sum over vectors
that have e.g.\ $2k$ $1$s in the first part of $e$. We can repair this by adding
the expression
\begin{equation} \label{eq:w1:4}
   \alpha \cdot \prod_{i = 0}^{k - 1} \left( i - \sum_{j = 1}^{p(n)} e_i \right)
   \prod_{j = 0}^{k - 1} \left( j - \sum_{j = p(n) + 1}^{p(n) + q(n)} e_i \right),
\end{equation}
which is nonzero if there are $k$ $1$s in the first and $k$ $1$s in the second part
and $0$ otherwise. By choosing $\alpha$ appropriately, we can achieve that this expression
is $1$ in the first case. The expression in (\ref{eq:w1:4}) has constant depth and weft 2,
therefore, we obtain:

\begin{proposition}\label{prop:w1}
If $(f_n(X_1,\dots,X_n,Y_1,\dots,Y_{p(n)}, Z_1,\dots,Z_{q(n)}))$ is computed
by polynomial size constant depth circuits of weft $t \ge 2$, then
the family $(g_{n,k})$ given by (\ref{eq:w1:3}) is in $\VW t$.
\end{proposition}

The case $t = 1$ is more complicated. We will also need these techniques later on the $\VW 1$-completeness
proof of the clique family. We start with a technical lemma.

\begin{lemma}\label{lem:w1:1}
Let $k+1$ be prime. The system of equations
\begin{align*}
   a \cdot b & = c \\
   a + b + c & = k^2 + 2k
\end{align*}
has a unique integer solution $a,b,c \ge 1$, namely, $a = b = k$ and $c = k^2$.
\end{lemma}

\begin{proof}
Plugging the first equation in to the second, we get
\[
   (a + 1)(b + 1) = a + b + ab  + 1 = k^2 + 2k + 1 = (k + 1)^2.
\]
Since $a + 1, b + 1 \ge 2$ and $k+1$ is prime, the only possible solutions
are $a = b = k$.
\end{proof}

The next construction gives a solution to the composition problem when $k+1$ is prime.

\begin{lemma}\label{lem:w1:2}
Let $n \ge  k \ge 1$
be integers such that $k + 1$ is prime. There is a Boolean expression $B$ of size $\poly(n)$ and
weft~1 on $n^2 + 2n$
variables $X_1,\dots,X_n$, $Y_1,\dots,Y_n$, and $Z_{i,j}$, $i,j = 1,\dots,n$
such that:
\begin{itemize}
\item Every satisfying assignment that sets $k^2 + 2k$ variables to $1$
sets $k$ variables from $X_1,\dots,X_n$ to $1$,
$k$ variables of $Y_1,\dots,Y_n$ to $1$, and $k^2$ many variables from $Z_{i,j}$ to $1$.
\item For every partial assignment that sets $k$ variables of $X_1,\dots,X_n$ to $1$,
and $k$ variables of $Y_1,\dots,Y_n$ to $1$, there is a unique extension of this
partial assignment to a satisfying assignment of $B$. This satisfying assignment
sets $k^2$ variables of the $Z_{i,j}$ to $1$.
\end{itemize}
\end{lemma}

\begin{proof}
Consider the complete bipartite graph with $n$ nodes on each side. The variables $X_1,\dots,X_n$
represent the nodes on the one side, the variables $Y_1,\dots,Y_n$ represent
the nodes on the other side and the variables $Z_{i,j}$ represent the edges.
Setting a variable to $1$ means that we select the corresponding node or edge.

The Boolean formula $B$ expresses the fact that we select a complete bipartite subgraph, that is:

$B$ is given by
\[
   \bigwedge_{i,j = 1}^n (Z_{i,j} \leftrightarrow X_i \wedge Y_j) \wedge \bigvee_{i = 1}^n X_i \wedge \bigvee_{i = 1}^n Y_i,
\]
i.e., an edge is selected if and only if both endpoints are selected. The two ORs ensure that
we select at least one node on each side. Consider an assignment
that sets in total $k^2 + 2k$ variables to $1$,
$a$ variables from $X_1,\dots,X_n$ to $1$, $b$ variables from $Y_1,\dots,Y_n$ to $1$,
and $c$ variables of the $Z_{i,j}$ to $1$. That is, we have $a + b + c = k^2 + 2k$.
If this assignment is satisfying, then it encodes a complete bipartite subgraph and
we therefore have $a b = c$. Since $k+1$ is prime, it follows from
\cref{lem:w1:1} that $a = b = k$ and $c = k^2$.

For the second part, note that the partial assignment selects $k$ nodes on both sides.
If we choose the $k^2$ edges between these two sets, we get a complete subgraph, and this
is the only way.
\end{proof}

\Cref{lem:w1:2} can also be extended to arbitrary values of $k$. In this case, the size of $B$ also depends
on $k$ but only polynomially.

\begin{lemma}\label{lem:w1:3}
Let $n \ge k \ge 1$ be integers. There are integers $k \le \ell \le 2k$
and $m \le n + k$
and a Boolean expression $B$ of size $\poly(n,k)$ and weft~1 on $m^2 + 2m$
variables $X_1,\dots,X_m$, $Y_1,\dots,Y_m$, and $Z_{i,j}$, $i,j = 1,\dots,m$
such that:
\begin{itemize}
\item Every satisfying assignment that sets $\ell^2 + 2\ell$ variables to $1$
sets $k$ variables from $X_1,\dots,X_n$ to $1$, $\ell - k$ variables from
$X_{n + 1},\dots,X_m$ to $1$, $k$ variables from $Y_1,\dots,Y_n$ to $1$ and
$\ell - k$ variables from $Y_{n + 1},\dots,Y_m$ to $1$, and $\ell^2$ variables
of the $Z_{i,j}$ to $1$.
\item For every partial assignment that sets $k$ variables of $X_1,\dots,X_n$ to $1$,
and $k$ variables of $Y_1,\dots,Y_n$ to $1$, there is a unique extension of this
partial assignment to a satisfying assignment of $B$. This satisfying assignment
sets $\ell$ variables from $X_1,\dots,X_m$ and $Y_1,\dots,Y_m$ to $1$,
respectively, and $\ell^2$ variables of the $Z_{i,j}$ to $1$.
\end{itemize}
\end{lemma}

\begin{proof}
Choose the smallest prime $\ell \ge k$. By Bertrand's postulate, $\ell \le 2k$.
Let $m = n + \ell - k$.
Let $B'$ be the expression constructed in \cref{lem:w1:2} with parameters $m$ and $\ell$ (in place of $n$ and $k$).
The expression $B$ is defined by
\[
   B = B' \wedge \bigwedge_{i = n + 1}^{n + \ell - k} (X_i \wedge Y_i).
\]
The expression $B'$ ensures that $\ell$ variables from $X_1,\dots,X_m$ and $Y_1,\dots,Y_m$
are set to $1$ and the remaining part ensures that we set exactly $k$
variables  from $X_1,\dots,X_n$ and $Y_1,\dots,Y_n$ to $1$.

For the second part, note that to satisfy $B$, we are forced to set
$X_{n+1},\dots,X_m$ and $Y_{n+1},\dots,Y_m$ to $1$.
\end{proof}

We will also need the following generalization of the previous lemma:

\begin{lemma}\label{lem:w1:4}
Let $n_1$ and $n_2$ be integers and $n = \max\{n_1,n_2\}$. Let $s\colon \Nset \to \Nset$
a function with $s(k) \ge k$ for all $k$.
Let $n \ge k \ge 1$ be an integer such that $s(k) \le n_2$.
There are integers $k \le \ell \le 2s(k)$ and $m \le n + 2s(k)$
and a Boolean expression $B$ of size $\poly(n,s(k))$ and weft~1 on $m^2 + 2m$
variables $X_1,\dots,X_m$, $Y_1,\dots,Y_m$, and $Z_{i,j}$, $i,j = 1,\dots,m$
such that:
\begin{itemize}
  \item Every satisfying assignment that sets $\ell^2 + 2\ell$ variables to $1$
sets $k$ variables from $X_1,\dots,X_{n_1}$ to $1$, $\ell - k$ variables from
$X_{n_1 + 1},\dots,X_m$ to $1$, $s(k)$ variables from $Y_1,\dots,Y_{n_2}$ to $1$ and
$\ell - s(k)$ variables from $Y_{n_2 + 1},\dots,Y_m$ to $1$, and $\ell^2$ variables
of the $Z_{i,j}$ to $1$.
\item For every partial assignment that sets $k$ variables of $X_1,\dots,X_{n_1}$ to $1$,
and $s(k)$ variables of $Y_1,\dots,Y_{n_2}$ to $1$, there is a unique extension of this
partial assignment to a satisfying assignment of $B$. This satisfying assignment
sets $\ell$ variables from $X_1,\dots,X_m$ and $Y_1,\dots,Y_m$ to $1$,
respectively and $\ell^2$ variables of the $Z_{i,j}$ to $1$.
\end{itemize}
\end{lemma}

\begin{proof}
The construction is similar to the one of \cref{lem:w1:3}.
Choose the smallest prime $\ell \ge s(k)$. By Bertrand's postulate, $\ell \le 2s(k)$.
Let $m = n + \ell - k$.
Let $B'$ be the expression constructed in \cref{lem:w1:2} with parameters $m$ and $\ell$.
The expression $B$ is defined by
\[
   B = B' \wedge \bigwedge_{i = n_1 + 1}^{n_1 + \ell - k} X_i \wedge
                 \bigwedge_{i = n_1 + \ell - k + 1}^{n + \ell - k} \neg X_i \wedge
                 \bigwedge_{i = n_2 + 1}^{n_2 + \ell - s(k)} Y_i \wedge
                 \bigwedge_{i = n_2 + \ell - s(k) + 1}^{n + \ell -k} \neg Y_i.
\]
Now the second part of the expression $B$ ensures that $k$ variables of $X_1,\dots,X_{n_1}$
are set to $1$ and $s(k)$ of $Y_1,\dots,Y_{n_2}$ are set to $1$. Since we have different
values $n_1$ and $n_2$, some of the $X_i$ or $Y_i$ will be set to $0$.
\end{proof}

Now we can also extend \cref{prop:w1} to $t = 1$.

\begin{theorem}\label{thm:sum:clique}
Let $(f_n(X_1,\dots,X_n,Y_1,\dots,Y_{p(n)}, Z_1,\dots,Z_{q(n)}))$ be computed
by polynomial size constant depth circuits of weft $t \ge 1$ and
Let $s\colon \Nset \to \Nset$ with $s(k) \ge k$ for all $k$. Then
the family $(g_{n,k})$ defined by
\[
   g_{n,k} = \sum_{d \in \ones {p(n)} k} \sum_{e \in \ones {q(n)} {s(k)}}
            f_n(X_1,\dots,X_n,d_1,\dots,d_{p(n)},e_1,\dots,e_{q(n)})
\]
is in $\VW t$.
\end{theorem}

\begin{proof}
Let $B$ be an arithmetization\footnote{By an arithmetization of a Boolean formula
we mean an arithmetic formula that is obtained by replacing every AND-gate by a multiplication
and every negation $\neg X$ by $1-X$. OR-gates are replaced using de Morgan's law.
When we restrict the inputs of the arithmetized formula to $\{0,1\}$, then it behaves like the Boolean one.} of
the expression from \cref{lem:w1:4} with $n_1 \coloneqq p(n)$,
$n_2 \coloneqq q(n)$. $B$ has weft~1.
Let $m$ and $\ell$ be defined as in \cref{lem:w1:4},
$B$ has $m^2 + 2m$ variables. We can write
\[
   g_{n,k} = \sum_{e \in \ones{m^2 + 2m}{\ell^2 + 2\ell}} B(e)
    f_n(X_1,\dots,X_n,e_1,\dots,e_{p(n)},e_m,\dots,e_{m + q(n)}).
\]
The properties of $B$ ensure that only the vectors $e$ survive for which $(e_1,\dots,e_{p(n)})$ and
$(e_m,\dots,e_{m + q(n)})$ have $k$ and $s(k)$ $1$s, respectively,
and that for each such pair of vectors there is exactly one such $e$.
\end{proof}

\subsection{Weft 1 Arithmetic Circuits}\label{app:weft}

We prove a normal form for circuits of weft 1.

\begin{lemma}\label{lem:struct:0}
Let $C$ be a circuit of depth $d$, size $s$ and weft $w$.
Then there is a formula of size $O(2^{d - w}s^w)$ computing the same functions as $C$.
\end{lemma}

\begin{proof}
We use the standard technique of duplicating gates whenever they are used more than once:
Whenever we encounter a gate $g$ that has fan-out $o > 1$, then
we create a total of $o$ copies of the subcircuit induced by $g$. When this process stops,
we end up with a formula computing the same polynomial. Since we only create copies of the subcircuits,
the depth and the weft do not change.

Any path of the resulting formula from the root to a leaf has $d-w'$ gates of fan-in~$2$
and $w'$ gates of unbounded fan-in---which can be at most $s$---on it for some $w' \le w$.
From this, the size bound follows.
\end{proof}

\begin{lemma}\label{lem:struct:1}
Let $(C_{n,k})$ be a family of circuits of constant depth and fan-in~$2$. Then each
$C_{n,k}$ computes a polynomial of constant degree that depends on a constant number
of variables.
\end{lemma}

\begin{proof}
By \cref{lem:struct:0}, we can assume that $C_{n,k}$ is a formula of constant size,
constant depth, and fan-in~2. (Note that the weft of $C_{n,k}$ is $0$.)
Such formula has a constant number of leaves and a constant number of multiplication gates.
\end{proof}

\begin{lemma}\label{lem:struct:2}
Let $(C_{n,k})$ be a family of arithmetic circuits of fpt size, constant depth,
and weft~1. Then there is a family of arithmetic formulas $(F_{n,k})$ of fpt size and a constant $b$
such that:
\begin{enumerate}
\item The top gate of $F_n$ is an addition gate of fan-in $\le b$,
\item the second layer of $F_n$ are multiplication gates of fan-in $\le b$,
\item the third layer are addition and multiplication gates of unbounded fan-in,
\item the fourth layer are addition gates of fan-in $\le b$, and
\item the fifth layer are multiplication gates of fan-in $\le b$.
\end{enumerate}
Furthermore, every multiplication gate in the fifth layer has exactly one input that is labeled with
a constant and every addition gate of the fourth layer has at most one child that
computes a constant.
\end{lemma}

\begin{proof}
By \cref{lem:struct:0}, we can assume that $C_{n,k}$ is a formula.
The subcircuit $S_v$ that is induced by any child $v$ of an unbounded fan-in gate
is a formula of constant depth and bounded fan-in. By \cref{lem:struct:1},
$S$ computes a polynomial $p_v$ of constant degree depending on a constant number of variables.
Such a polynomial has only a constant number of monomials of constant degree. Each monomial
can be computed by a multiplication gate of bounded fan-in and these monomials can be added
by an addition gate of bounded fan-in. Each multiplication gate has one input which is not
a variable. There might be one monomial which is a nonzero constant.
For this monomial, we will still insert a multiplication
gate of fan-in $1$ for consistency reasons.

Consider the circuit $D$ that is obtained from $C_{n,k}$ by removing all subcircuits
induced by the gates of unbounded fan-in and replacing each gate $g$ by a new indeterminate $Z_g$.
The circuit $D$ computes a polynomial $q$ of constant degree depending on a constant number of variables.
Like above, this polynomial can be computed by a summation gate followed by a layer of multiplication gates,
all having bounded fan-in.

We get the claimed formula $F_{n,k}$ by taking the depth-2-formula computing $q$ and then replacing
every variable $Z_g$ by the corresponding gate $g$ and every child $v$ of $g$ by the depth-2-formula
computing $p_v$.

There might be paths in $F_{n,k}$ from a leaf $\ell$
to the root that have length $2$, this happens when there is no unbounded fan-in gate on the path
from $\ell$ to the root in $C_{n,k}$. We can make this path longer by adding dummy gates of
fan-in $1$ to it.
\end{proof}

\subsection{Boolean-arithmetic Formulas}\label{app:bool}

A \emph{Boolean-arithmetic formula} is a formula of the form
\[
    B(X_1,\dots,X_n) \cdot \prod_{i = 1}^n (R_i X_i + 1 - X_i)
\]
where $B$ is an arithmetization of some Boolean formula and the $R_i$ some polynomial
or even rational function (over a different set of variables). For each satisfying $\{0,1\}$-assignment $e$
to $B$, the right hand side produces one product  and the $e_i$ switch the factors $R_i$ on or off.
Boolean-arithmetic formulas will play an important role for the hardness proof of the clique family. Note that if $B$
has weft~1, then the whole formula has weft 1, since the right hand side has weft~1.

The \emph{support} of a monomial $m$ is the set of variables appearing in $m$ (with positive degree).
The \emph{support size} of $m$ is the number of variables in its support. For some polynomial $f$,
$\spc k(f)$ is the sum of all monomials in $f$ that have support size equal to $k$. Like for the homogeneous parts,
we have that $f = \sum_{i = 0}^s \spc k(f)$, where $s$ is the largest support size of any monomial.
Note that for any monomial, the support size is always bounded by the degree and we have equality iff
the monomial is multilinear.

The aim of this section is to prove that for a bounded depth arithmetic circuit $C$ of weft~1 computing a polynomial $f$ in $n$ variables,
we can write the part with support size $k$ of $f$ as a bounded sum
over a Boolean arithmetic expression.
In this section, $F$ is a formula as in \cref{lem:struct:2} computing the polynomial $f \in K[X_1,\dots,X_n]$.
For $A \subseteq \{X_1,\dots,X_n\}$, we consider the subtree that is obtained from
$F$ by selecting all monomials $m$ that are computed at the multiplication gates of the
fifth layer such that the support of $m$ is contained in $A$. We denote this subformula
by $F|_A$ and the polynomial computed by it by $f|_A$.

\begin{observation}
$f|_A$ is the sum of all monomials of $f$ that contain only variables from $A$.
$f|_A$ is obtained from $f$ by setting all variables not in $A$ to $0$.
\end{observation}

By the last observation, we can view $|_A$ as an operator that can be applied independently
from a given formula $F$.

\begin{observation}
The operator $|_A$ is additive and multiplicative, that is, for polynomials $f$ and $g$, we have
\begin{enumerate}
\item $(f+g)|_A = f|_A + g|_A$ and
\item $(fg)|_A = (f|_A)(g|_A)$.
\end{enumerate}
\end{observation}

\begin{lemma}\label{lem:S:inclusion}
For any $A \subseteq \{X_1,\dots,X_n\}$, $\lvert A\rvert = k$, we have
\[
   \spc {k}(f|_A) = \sum_{B \subseteq A} (-1)^{k - \lvert B\rvert} f|_B.
\]
\end{lemma}

\begin{proof}
Let $A = \{X_{i_1},\dots,X_{i_k}\}$ and let $A_j = A \setminus \{X_{i_j}\}$ for all $1\le j \le k$.
$f|_A$ only has monomials with support size at most $k$.
$\spc k(f|_A)$ are all monomials of $f|_A$ that have exactly support size $k$.
The monomials of smaller support size are the monomials of $f|_{A_1},\dots,f|_{A_k}$.
Now the claim follows by the inclusion-exclusion principle. 
\end{proof}

\begin{lemma}\label{lem:spcfA}
For any $0 \le k \le n$, we have
\[
   \spc k(f) = \sum_{\ell = 0}^k (-1)^{k - \ell} \binom{n - \ell}{k - \ell}
               \sum_{\substack{A \subseteq \{X_1,\dots,X_n\}\\ {\lvert A\rvert = \ell}}} f|_A.
\]
\end{lemma}

\begin{proof}
We have
\[
   \spc k(f) = \sum_{\substack{A \subseteq \{X_1,\dots,X_n\}\\ {\lvert A\rvert = k}}} \spc k(f|_A),
\]
since every monomial on the right hand side appears exactly once one the left hand side and vice
versa. By \cref{lem:S:inclusion}, we get
\[
   \spc k(f) = \sum_{\substack{A \subseteq \{X_1,\dots,X_n\}\\ {\lvert A\rvert = k}}} \sum_{B \subseteq A} (-1)^{k - \lvert B\rvert } f|_B.
\]
If $\lvert B\rvert = \ell$, then $B$ is a subset of $\binom {n - \ell}{k - \ell}$ different $A \subseteq \{X_1,\dots,X_n\}$ of size $k$.
Therefore,
\[
   \spc k(f) = \sum_{\ell = 0}^k (-1)^{k - \ell} \binom{n - \ell}{k - \ell}
               \sum_{\substack{A \subseteq \{X_1,\dots,X_n\}\\ {\lvert A\rvert  = \ell}}} f|_A. \qquad \qed
\]

\end{proof}

This means in order to compute $\spc k(f)$, it is sufficient to compute the sums
\[
   \sum_{\substack{A \subseteq \{X_1,\dots,X_n\}\\ {\lvert A\rvert = \ell}}} f|_A.
\]
This is a bounded summation, so we need to understand how to compute $f|_A$ given some particular $A$.
Let $v$ be a summation gate in the fourth layer of $F$ computing a polynomial $p_v$.
Let $T_v$ be the support of $p_v$, which is the union of the supports of all monomials in $p_v$.

\begin{observation}
Given some $A \subseteq \{X_1,\dots,X_n\}$ of size $k$. Then $p_v|_A = p_v|_{A \cap T_v}$.
\end{observation}

\begin{lemma}\label{lem:restrict:1}
Let $u$ be a multiplication gate in the third layer computing a polynomial $q$.
Then
\[
  q|_A = \prod_{v} p_v|_{A \cap T_v}
\]
where the product is taken over all children $v$ of $u$ in the fourth layer
and $p_v$ is the polynomial computed at $v$.
\end{lemma}

\begin{proof}
We have $q = \prod_v p_v$. Since $|_A$ is multiplicative,
$q|_A = \prod_v p_v|_A$. By the observation above, we can replace $A$ by $A \cap T_v$.
\end{proof}

In the same way, we can prove:

\begin{lemma}\label{lem:restrict:2}
Let $u$ be an addition gate in the third layer computing a polynomial $q$.
Then
\[
  q|_A = \sum_{v} p_v|_{A \cap T_v}
\]
where the sum is taken over all children $v$ of $u$ in the fourth layer and $p_v$ is the polynomial
computed at $v$. \qed
\end{lemma}

Let $u$ be a multiplication gate in the second layer. There is only a constant number of
such gates, therefore, it will turn out that we can treat them separately. We want to
write the polynomial $q$ computed at $u$ as a bounded sum over a Boolean-arithmetic formula.
Let $v_1,\dots,v_s$ and $w_1,\dots,w_t$ be the
children of $u$, where $v_1,\dots,v_s$ are unbounded fan-in multiplication gates and
$w_1,\dots,w_t$ are unbounded fan-in addition gates. The polynomial $q$ computed at $u$ is
\[
   q = \prod_{i = 1}^s \prod_{\text{$c$ child of $v_i$}} p_c \cdot \prod_{i = 1}^t \sum_{\text{$c$ child of $w_i$}} p_c.
\]
By \cref{lem:restrict:1,lem:restrict:2}, we have
\begin{equation} \label{eq:BA}
   q|_A = \prod_{i = 1}^s \, \prod_{\text{$c$ child of $v_i$}} p_c|_{A \cap T_c} \cdot \prod_{i = 1}^t \, \sum_{\text{$c$ child of $w_i$}} p_c|_{A \cap T_c}.
\end{equation}

When we write $q|_A$ as a bounded sum over a Boolean-arithmetic expression, we write it essentially as
a bounded sum over a product. This means that from the first outer product of the expression in (\ref{eq:BA}), every
factor can potentially contribute and from the second outer product, we choose one summand from
each sum and sum over all possibilities. There is one complication: Since the sum is bounded,
we are only allowed to set a number of Boolean variables to $1$ that is bounded by some
function in $k \coloneqq \lvert A\rvert$. For instance, $\prod_{i = 1}^n X_1$ has support size $1$, however, the degree is unbounded.

We will have $n$ Boolean variables $y_1,\dots,y_n$. Setting such a variable $y_\nu$  to $1$ indicates
that the algebraic variable $X_\nu$ is in the support $A$ that we are currently considering. In the following $a$
will be an Boolean assignment to the $y_1,\dots,y_n$ such that assigns the value $1$ to exactly $k$ variables.
We will extend this assignment to further variables. Once the values to the $y_\nu$ are fixed, there will only
very few ways to extend the assignment and we will have full control over it.

In our construction, we will treat the gates $v_i$ and $w_j$ differently. From each $w_j$ (addition gate), only one child can contribute at a time.
For each $v_i$ (multiplication gate), all children can contribute simultaneously.
Let $c$ be a child of one of the $v_i$.
For every subset $B$ of $T_c$, we will have a Boolean variable $x_B$.
If a set $B$ is a subset of more than one $T_c$, there will only be one variable $x_B$.
For every $x_B$, we will have a Boolean expression that sets $x_B$ to $1$ iff
$B \subseteq \{X_{i} \mid a(y_i) = 1, 1 \le i \le n\}$. We can easily express this by a Boolean expression
of weft 1:
\begin{equation}
   \cB_1 = \bigwedge_B  \Bigl(x_B \leftrightarrow \bigwedge_{i:\, X_i \in B} y_i  \Bigr).
\end{equation}
Here the first AND is over all $B$ that appear in some $T_c$. Note that the inner
AND gate is bounded. The number of variables that will be set to $1$ by this
expression is at most $2^k$.

For every $w_j$ and every child $c$ of $w_j$, we have a ``switch'' variable $s_{j,c}$.
For each $j$, exactly one of them shall be set to $1$. This corresponds to choosing
the corresponding summand $c$. Let $\cB_2$ be a weft 1 Boolean expression
realising this, that is
\[
   \cB_2 = \bigwedge_{j = 1}^t \Bigl( \Bigl[ \bigvee_{\text{$c$ child of $w_j$}} s_{j,c} \Bigr]
            \wedge \neg \Bigl[ \bigwedge_{\text{$c \not= c'$ children of $w_j$}} (s_{j,c} \wedge s_{j,c'}) \Bigr]  \Bigr).
\]

Now let $c$ be a child of one of the $w_j$. For each such $c$ and every $B \subseteq T_c$,
we have a variable $z_{c,B}$. For every $z_{c,B}$ there will be a Boolean expression $\cB_{3,c,B}$
that sets $z_{c,B}$ to $1$ if $s_{j,c}$ is set to $1$ and $B \subseteq \{X_{i} \mid a(y_i) = 1, \enspace 1 \le i \le n\}$.
Note that $\cB_{3,c,B}$ has constant size, since $T_c$ has, and therefore
\[
  \cB_3 = \bigwedge_{j = 1}^t \bigwedge_{\text{$c$ child of $w_j$}} \bigwedge_{B \subseteq T_c} \cB_{3,c,B}
\]
has weft $1$, since the AND gate in the middle is the only one that has unbounded fan-in.
Furthermore, note that we can afford to have a separate variable $z_{c,B}$ for every $c$,
since in a summation gate, only one child is selected and therefore, we will only
set a number of variables to $1$ that is bounded by $2^{k}$.

\begin{lemma}\label{lem:moebius}
For any polynomial $p$ and any set $B$, let
\[
   Q_B = \prod_{C \subseteq B} (p|_C)^{{(-1)}^{\lvert B\rvert - \lvert C\rvert }},
\]
where we assume that all $p|_C$ are nonzero.
(The factors are either $p|_C$ or $(p|_C)^{-1}$, depending on the parity of $\lvert B\rvert - \lvert C\rvert$.)
Then for any $A$,
\[
   \prod_{B \subseteq A} Q_B = p|_A.
\]
\end{lemma}

\begin{proof}
We can write
\[
   \prod_{B \subseteq A} Q_B = \prod_{B \subseteq A} \prod_{C \subseteq B} (p|_C)^{{(-1)}^{\lvert B\rvert - \lvert C\rvert}}
\]
Fix a particular $C \not= A$. Let $m = \lvert A\rvert  - \lvert C\rvert$.
First assume that $m$ is odd. $p|_C$ appears $\sum_{i = 0}^{(m-1)/2} \binom{m}{2i}$ times in  $\prod_{B \subseteq A} Q_B$
and $(p|_C)^{-1}$ appears $\sum_{i = 0}^{(m-1)/2} \binom{m}{2i + 1}$ times. Both contributions cancel.
If $m$ is even, then $p|_C$ appear $\sum_{i = 0}^{m/2} \binom{m}{2i}$ times in  $\prod_{B \subseteq A} Q_B$
and $(p|_C)^{-1}$ appears $\sum_{i = 0}^{m/2 - 1} \binom{m}{2i + 1}$ times. Again both contributions cancel.
Only $p|_A$ remains.
\end{proof}

\begin{remark}\label{rem:moebius}
If some $p|_C$ is zero in the lemma above, then the lemma is still true when we omit every
$p|_C$ that is $0$ from every $Q_B$.
The only exception is when $p|_A = 0$. Then the expression will be $1$.
\end{remark}

We want to use expressions as in \cref{lem:moebius} to design our Boolean arithmetic formula.
However, these expressions give the wrong result when $p|_A$ is zero.\footnote{Note that for a polynomial,
if $p|_B \not= 0$, then $p|_{B'} \not= 0$ for all $B' \supseteq B$. But we will later consider
the case when instead of $p|_B$, we will consider $p|_B$ evaluated at some point. Therefore, we will not make
use out of this property.}
We will have a Boolean variable $w$
that is set to $1$ when for some child $c$ of some $v_i$, $p_c|_A$ is $0$. This is achieved
by the following expression:
\[
   \cB_4 = \left(\bigvee_{i = 1}^s \bigvee_{\text{$c$ child of $v_i$}} \bigvee_{\substack{B \subseteq T_c\\  {p_c|_B = 0}}} (\bigwedge_{i \in B} y_i \wedge \bigwedge_{i \in T_c \setminus B} \neg y_i) \right) \leftrightarrow w.
\]
$\cB_4$ has weft $1$, since all gates have constant fan-in except the second one. Note that we can precompute whether $p_c|_B = 0$.

Let $a$ be some assignment that sets $k$ variables $y_1,\dots,y_n$ to $1$, defining some set $A$
of size $k$. We extend $a$ to an assignment to the other variables
such that $\cB_1 \wedge \cB_2 \wedge \cB_3 \wedge \cB_4$ is satisfied. The number of variables
set to $1$ is bounded by a function in $k$.
Note that the only degree of
freedom we have is to choose which $s_{j,c}$ we set to $1$. Once this choice has been made,
there is only one way to extend the partial assignment to a satisfying one.

For some child $c$ of some $v_i$ and $B \subseteq T_c$, let
\[
   Q_{c,B} = \prod_{C \subseteq B} (p_c|_C)^{(-1)^{\lvert B\rvert - \lvert C\rvert}}.
\]
If some $p_c|_C$ are zero, we simply leave them out. Consider the first outer product of (\ref{eq:BA}).
We claim that
\begin{equation} \label{eq:ba:new:1}
    \prod_{i = 1}^s \, \prod_{\text{$c$ child of $v_i$}} p_c|_{A \cap T_c}
      = (1 - a(w)) \prod_{\substack{B: \, \text{$\exists c$ with}\\ \text{$B \subseteq T_c$}} }
        \Bigl( a(x_B) \Bigl[ \prod_{i = 1}^s \prod_{\substack{\text{$c$ child of $v_i$}\\ {B \subseteq T_c}}} Q_{c,B} \Bigr] + 1 - a(x_B) \Bigr),
\end{equation}
where the outer product on the right-hand side is over all $B$ such that there is a child $c$ of some $v_i$ and $B \subseteq T_c$.
To see this, note that when all $p_c|_A = p_c|_{T_c \cap A} \not= 0$, then
\begin{align*}
 \prod_{\substack{B: \, \text{$\exists c$ with}\\ \text{$B \subseteq T_c$}} }
        \Bigl( a(x_B) \Bigl[ \prod_{i = 1}^s \prod_{\substack{\text{$c$ child of $v_i$}\\ {B \subseteq T_c}}} Q_{c,B} \Bigr] + 1 - a(x_B) \Bigr)
    & = \prod_{\substack{B: \, \text{$\exists c$ with}\\ \text{$B \subseteq T_c$ and $B \subseteq A$}} }
             \prod_{i = 1}^s \prod_{\substack{\text{$c$ child of $v_i$}\\ {B \subseteq T_c}}} Q_{c,B}  \\
    & = \prod_{i = 1}^s \prod_{\text{$c$ child of $v_i$}} \prod_{B \subseteq A \cap T_c} Q_{c,B} \\
    & = \prod_{i = 1}^s \prod_{\text{$c$ child of $v_i$}} p_c|_{A \cap T_c},
\end{align*}
where the last equality follows
by \cref{lem:moebius} and the remark after the lemma. Recall that $a(x_B)$ is true if $a$ sets all variables from $y_\nu$ to $1$ for which $X_\nu \in B$.
When $p_c|_A = p_c|_{T_c \cap A} = 0$
for some $c$, then the whole product shall be zero. This is achieved by the first factor $(1 - a(w))$.

We write this rather complicated expression to get a Boolean-arithmetic formula.
Since the outer product in (\ref{eq:ba:new:1}), we do not take product over all $B$, but only
over subsets of all $T_{c}$, so the product has polynomial size, since each $T_{c}$ has constant size.

For the addition gates, we can do something similar. The construction is even easier. We can write
\begin{equation} \label{eq:ba:new:2}
   \prod_{i = 1}^t \, \sum_{\text{$c$ child of $w_i$}} p_c|_{A \cap T_c}
      = (1 - a(w')) \prod_{i = 1}^t \prod_{\text{$c$ child of $w_i$}} \prod_{B \subseteq T_c} \Bigl( a(z_{c,B}) Q_{c,B} +
                      1 - a(z_{c,B}) \Bigr).
\end{equation}
Here, the product is not taken over all $B$ but again only over subsets of $T_c$. $w'$ is again a new variable
that checks whether one of the chosen $p_c|_A$ is zero. The corresponding Boolean expression is given by
\[
   \cB_5 = \left(\bigvee_{j=1}^t \bigvee_{\text{$c$ child of $w_j$}} \bigvee_{\substack{B \subseteq T_c\\ {p_c|_B = 0}}} \Bigl[ z_{c,b} \wedge (\bigwedge_{i \in B} y_i \wedge \bigwedge_{i \in T_c \setminus B} \neg y_i) \Bigr] \right) \leftrightarrow w'.
\]

Now we can write
\begin{align}
 \sum_{\substack{A \subseteq \{X_1,\dots,X_n\}\\ {\lvert A\rvert = k}}} q|_A & = \sum_a \cB(a) (1 - a(w))
    \prod_B \Bigl( a(x_B) \prod_{i = 1}^s \prod_{\text{$c$ child of $v_i$}}  a(x_B) Q_{c,B} + 1 - a(x_B) \Bigr)  \notag \\
  & \qquad \qquad \quad (1-a(w')) \prod_{i = 1}^t \prod_{\text{$c$ child of $w_i$}} \prod_{B \subseteq T_c} \Bigl( a(z_{c,B}) Q_{c,B} + 1 - a(z_{c,B}) \Bigr). \label{eq:lalala}
\end{align}
Here the sum is taken over all $a$ that set $k$ of the $y_i$ to $1$ and its extension satisfies the expression $\cB \coloneqq \cB_1 \wedge \dots \wedge \cB_5$.
Every such assignment sets at most $q(k)$ other variables to $1$ for some function $q$.
To make it the same number for all assignments $a$,
we add $q(k)$ dummy variables $d_1,\dots,d_{q(k)}$, which can always be set to $1$
(but which do not affect the formula $\cB$). However, there should be only one unique way to
extend a satisfying assignment of $\cB$ to a satisfying assignment with exactly $q(k)$ $1$s.
Therefore, we add the expression $\bigwedge_{i = 1}^{q(k)} (d_i \ge d_{i+1})$ to $\cB$,
which has weft 1. Now the sum over all $a$ that set $k$ variables from $y_1,\dots,y_n$ to $1$
and $q(k)$ other variables
is a doubly bounded sum and we can combine it into one by \cref{lem:w1:4}.

Finally, we can write $f|_A$ as a constant sum over constantly many $q|_A$.
By \cref{lem:spcfA}, we can write $\spc k (f)$ as a bounded sum over $f|_A$.
By applying the lemma below for a bounded number of times (in $k$), we obtain \cref{thm:main:1:yyyy}.

\begin{lemma}\label{lem:sum:ba}
Let $f$ and $g$ be computed by bounded sums over Boolean-arithmetic formulas of weft 1, that is,
\[
   f = \sum_{d \in \ones m k} B(d_1,\dots,d_m)  \prod_{i = 1}^m (R_i d_i + 1 - d_i) \quad \text{and} \quad
   g = \sum_{e \in \ones n k} C(e_1,\dots,e_n)  \prod_{j = 1}^n (M_j e_j + 1 - e_j).
\]
Then $f + g$ can also be expressed by a bounded sum over a Boolean-arithmetic formula of weft 1.
\end{lemma}

\begin{proof}
We can write
\begin{align*}
    f + g & = \sum_{d \in \ones nk} \sum_{e \in \ones m{k}} \sum_{s \in \{0,1\}}
               \left( s \wedge \left(B(d_1,\dots,d_m) \wedge \bigwedge_{j = 1}^n \neg e_j\right) \vee \bar s \wedge \left(\bigwedge_{i = 1}^m \neg d_i \wedge C(e_1,\dots,e_n)\right)\right)
               \\
     & \qquad \quad \cdot \prod_{i = 1}^m (R_i d_i + 1 - d_i) \prod_{j = 1}^n (M_j e_j + 1 - e_j).
\end{align*}
The Boolean expression is only true if the $d_i$ or $e_i$ are all $0$. In this case, the corresponding algebraic product is $1$.
The switch $s$ selects which expression currently contributes.
Now we can merge the three sums into one by \cref{lem:w1:4}.
\end{proof}

\begin{theorem}\label{thm:main:1:yyyy}
Let $C$ be a circuit of bounded depth and weft~1 computing a polynomial $f$ in $n$ variables.
Then there is a Boolean formula $\cB$ of weft 1 and fpt size (in $k$ and the size of $C$) having $p(n,k)$
inputs for some fpt-bounded function $p$ and rational functions $R_i$ such that for all $k$,
\begin{equation} \label{eq:last1}
   \spc k(f) = \sum_{e \in \ones {p(n,k)}{q(k)}}
               \cB (e) \cdot \prod_{i = 1}^{p(n,k)} (R_i e_i + 1 - e_i)
\end{equation}
for some function $q$.
\end{theorem}

\begin{corollary}\label{cor:main:1:yyyy}
Let $C$ be a circuit of bounded depth and weft~1  computing a polynomial $f$ in $n$ variables.
Then there is a Boolean formula $\cB$ of weft 1 and fpt size (in $k$ and the size of $C$) having $p(n,k)$
inputs for some fpt-bounded function $p$ and rational functions $R_i$ with $R_i(1,\dots,1) \not= 0$
such that for all $k$,
\begin{equation} \label{eq:last2}
   \spc k(f)(1,\dots,1) = \sum_{e \in \ones {p(n,k)}{q(k)}}
               \cB (e) \cdot \prod_{i = 1}^{p(n,k)} (R_i(1,\dots,1) e_i + 1 - e_i)
\end{equation}
for some function $q$.
\end{corollary}

\begin{proof}
We have to ensure that all $R_i$ are defined at $(1,\dots,1)$. The $R_i$ are of the form
$\prod_{C \subseteq B} (p_c|_C)^{(-1)^{\lvert B\rvert - \lvert C\rvert}}$. Now they are replaced by
$\prod_{C \subseteq B} (p_c|_C(1,\dots,1))^{(-1)^{\lvert B\rvert - \lvert C\rvert}}$. If one of the
$p_c|_C(1,\dots,1)$ becomes $0$, we can leave it out, as before by \cref{rem:moebius}.
This can only produce a wrong result when $p|_B(1,\dots,1)$ becomes $0$.
We can repair this as before, by an appropriate modifications of the expression $\cB_4$ and $\cB_5$,
namely, we replace the condition $p_c|_B = 0$ in the second OR gate by $p_c|_B(1,\dots,1) = 0$.
\end{proof}

\subsection{$\VW 1$ and $\Cl$}\label{app:cliq}

In this section, we finally prove that $\Cl$ is $\VW 1$-complete.
We first show
that bounded sums over Boolean arithmetic formulas can be written as evaluations
of the clique polynomial.

\begin{lemma}\label{lem:cl:w1}
Let $f$ be computed by a Boolean-arithmetic expression
\begin{equation}\label{eq:cl:0}
   f(X_1,\dots,X_n) = \sum_{e \in \ones {p}{k}} B(e) \prod_{i = 1}^p (R_i e_i + 1 - e_i)
\end{equation}
where $B$ has size $s \ge p$ and weft~1.
Then there is a graph $G$ with adjacency matrix $A$ of size
$a(k)q(s)$ for some p-bounded function $q$ and an arbitrary
function $a$ and some $k' \le b(k)$ for some arbitrary function $b$
such that
\begin{equation}\label{eq:cl:00}
   f(X_1,\dots,X_n) = \Cl_{a(k) q(s),k'}(A,R_1,\dots,R_p,1,\dots,1).
\end{equation}
\end{lemma}

\begin{proof}
From the $\sharpW 1$-completeness proof~\cite{DBLP:series/txtcs/FlumG06, DBLP:series/txcs/DowneyF13}
of the parameterized clique problem, it follows
that there is a graph $G$ with adjacency matrix $A$ such that the satisfying assignments
of $B$ with $k$ $1$s stand in one-to-one correspondence with cliques of $G$ of size $k'$.
Furthermore, there is an injective mapping $I$ from the variables $x_1,\dots,x_p$ of the Boolean formula $B$ to the nodes
of $G$ such that whenever a variable $x_i$ is set to $1$ by a satisfying assignment $e$,
then $I(x_i)$ is contained in the clique corresponding to $e$. We give this node the weight $R_i$. We order
the nodes of $G$ such that these nodes are the first $p$ nodes. All other nodes will be labeled with $1$.
It follows that (\ref{eq:cl:0}) can be written as (\ref{eq:cl:00}), where we substitute the values of $A$ for the edge
variables and $R_1,\dots,R_p,1,\dots,1$ for the node variables.
\end{proof}

\begin{lemma}\label{lem:cl:spc}
Let $g \in K[Y_1,\dots,Y_m]$ and let
\[
   g_k = \sum_{e \in \ones m k} g(e).
\]
Then
\[
  g_k = \sum_{\ell = 0}^k \binom {m - \ell}{k - \ell} \spc \ell(g)(1,\dots,1).
\]
\end{lemma}

\begin{proof}
Let $\mu$ be a monomial of $g_k$. $\mu(e)$ contributes to $g_k$ if for all $Y_i$ in the support of $\mu$,
$e_i = 1$. If the support size of $\mu$ is $k$, then there is exactly one such vector $e$
and $\spc k(g)(1,\dots,1)$ is the sum over all $e$ of all $\mu(e)$ where $m$ has support size exactly $k$.
More general, if the support size of $\mu$ is $\ell \le k$, then there are $\binom {m - \ell}{k - \ell}$
such vectors.
\end{proof}

\begin{theorem}\label{thm:cl:w1}
$\Cl_{n,k}$ is complete for $\VW 1$ (under fpt-c-reductions).
\end{theorem}

\begin{proof}
Let $(f_{n,k})$ be a parameterized p-family in $\VW 1$. Let
$(g_n)$ be a p-family,
$g_{n} \in K[X_1,\dots,X_{p(n)},Y_1,\dots,Y_{q(n)}]$,
such that $g_{n}$ is computable by a constant depth unbounded fan-in circuit of weft $1$
and $(f_{n,k})$ is an fpt-substitution of the family $(g_{n,k})$ defined by
\begin{equation} \label{eq:cl:1}
   g_{n,k} = \sum_{e \in \ones{q(n)}{k} } g_{n}(X_1,\dots,X_{p(n)},e_1,\dots,e_{q(n)}),
\end{equation}
that is, there is a p-bounded function $m$, arbitrary functions $r,s$, and a $k' \le s(k)$
such that
\begin{align}
  f_{n,k} & = \sum_{e \in \ones{q(r(k)m(n))}{k'} } g_{r(k)m(n)}(h_1,\dots,h_{p(r(k)m(n))},e_1,\dots,e_{q(r(k)m(n))}) \notag \\
          & = g_{r(k)m(n),k'}(h_1,\dots,h_{p(r(k)m(n))})   \label{eq:cl:2}
\end{align}
for some polynomials $h_1,\dots,h_{p(r(k)m(n))}$ computable
by circuits of size $r(k)m(n)$ and of degree bounded by $r(k)m(n)$.

We now want to apply \cref{cor:main:1:yyyy} to the family $g_{n,k}$.
We interpret the $g_n$ and $g_{n,k}$ as a polynomial in $Y_1,\dots,Y_{q(n)}$ and treat
the other variables as constants (from the field $K(X_1,\dots,X_{p(n)})$).
By \cref{lem:cl:spc}, we can write
\[
   g_{n,k} = \sum_{\ell = 0}^k \binom {q(n) - \ell}{k - \ell} \spc \ell(g)(X_1,\dots,X_{p(n)},1,\dots,1).
\]

By \cref{cor:main:1:yyyy}, there are a Boolean expression $\cB_{n}$ of weft 1, an fpt-bounded
function $\hat p$ and some function $\hat q$ and
rational functions $R_{n,i}$ such that
\[
   \spc {k}(g_{n})(X_1,\dots,X_{p(n)},1,\dots,1)
   = \sum_{\hat e \in \ones {\hat p(n,k)}{\hat q(k)}} \cB_{n}(\hat e) \cdot
     \prod_{i = 1}^{\hat p(n,k)} (R_{n,i}(1,\dots,1) \hat e_i + 1 - \hat e_i).
\]
$R_{n,i}(1,\dots,1)$ is a constant. But here a constant is a
rational function from $K(X_1,\dots,X_{p(n)})$. However such a constant cannot be arbitrary.
The constants appearing in the proof
of \cref{cor:main:1:yyyy} are
computed by circuits of fpt size:
We start with building blocks $p_c|_C(1,\dots,1)$ or $1/p_c|_C(1,\dots,1)$, which are computed by constant size circuits,
and each $L_{n,i}$ might be an unbounded product or sum over them.  Each $p_c|_C(1,\dots,1)$
is a sum of polynomials in $X_1,\dots,X_{p(n)}$. These polynomials however were created when transforming the original
circuit $C$ of bounded depth and weft~1 into the normalized formula $F$. Therefore, they are itself computed
by circuits of constant depth.

Hence, we can write
\begin{equation} \label{eq:cl:3}
g_{n,k}(X_1,\dots,X_{p(n)}) = \sum_{\ell = 0}^k \binom {q(n) - \ell}{k - \ell}
               \sum_{\hat e \in \ones {\hat p(n,k)}{\hat q(k)}}
               \cB_{n}(\hat e) \cdot \prod_{i = 1}^{\hat p(n,k)} (R_{n,i}(X_1,\dots,X_{p(n)}) \hat e_i + 1 - \hat e_i).
\end{equation}
By \cref{lem:sum:ba}, we can replace the multiple sum in (\ref{eq:cl:3}) by
one sum. Note that the parameters in the theorem grow in such a way that the corresponding single sum
still has fpt size when we apply the theorem a number of times that is bounded by some function in $k$.
Therefore, we can interpret (\ref{eq:cl:3}) as a Boolean arithmetic formula
of weft 1.

By \cref{lem:cl:w1}, we can write $g_{n,k}$ as an evaluation of
$\Cl_{N,k'}$, where $N = a(k) \poly(n)$ and $k' \le b(k)$ for two functions $a,b$.
The expression that we plug into $\Cl_{N,k'}$ are $0$, $1$, and
the $R_{n,i}$. The $R_{n,i}$ are rational expressions that are computable by fpt size
circuits (in $n$ and $k$) and have fpt degree (in $n$ and $k$).\footnote{Note that from this together with
(\ref{eq:cl:2}) it already follows that $f_{n,k}$ is an fpt-substitution of
$\Cl$ when we allow that we substitute rational expression.}
Note that although we plug in rational functions into $\Cl_{N,k'}$,
the result is still a polynomial, namely $g_{n,k}$.

Assume for the moment that we were given a circuit computing
$\Cl_{N,k'}$ evaluated at the these rational functions.
Strassen~\cite{Strassen:JRAM73} shows that we can find a circuit without divisions.
His construction consists of the following steps:
\begin{itemize}
\item First we perform a Taylor shift of the variables such that all denominators
have nonzero constant term, that is, are invertible as formal power series.
\item Then we replace the denominators by approximations of high enough degree (still polynomial)
to the corresponding formal power series.
\item We compute with the homogeneous components up to the desired degree.
\item Finally, we revert the Taylor shift from the first step.
\end{itemize}
The difference in our setting is that we do not have a circuit but
we only have a substitution of a clique polynomial. So we cannot compute
in the third step with the homogeneous components. We overcome this difficulty
as follows:
\begin{itemize}
\item After the first two steps, we replace every variable $X_i$ by $TX_i$,
where $T$ is a new variable.
Let $P(T,X_1,\dots,X_n)$ be the resulting polynomial.
We view the polynomial as a univariate polynomial in $T$.
In this way, we group the homogeneous components of our polynomial.
\item Then we plug in enough (that is, fpt many) values for $T$ such that
we can recover the coefficients of $P$ by interpolation. Interpolation can
be done by polynomial sized circuits. In this way, we get
the homogeneous components of $P$ can continue like in Strassen's construction.
\end{itemize}

In this way, we get a circuit of fpt size computing $f_{n,k}$ using $\Cl$-oracle gates
\end{proof}

From the footnote to the proof of the above theorem, the following
corollary immediately follows:

\begin{corollary}
$\Cl$ is complete for $\VW 1$ under fpt-substitutions that are allowed to
substitute rational expressions.
\end{corollary}

\section{Hardness of the Permanent and Cycle Covers}\label{sec:permanent}
In this section we will highlight the vast difference between the provable complexity of the following two problems. Having cycle covers where one cycle is of length $k$ and all other cycles are self loops and the complexity of all cycle covers where one cycle is length $k$ and all other cycle covers are of length some fixed constant $c$. As always in this paper, we will look at corresponding polynomials to this problem.

We will extensively use that formulas are closed under extracting of homogeneous components. Let $K$ be a large enough field. Then given a polynomial $f=\sum_{i=1}^d x^i p_i(x_1,\dots,x_n)$ for some variable $x_i$, degree $d$ and polynomials $p_i\in K[x_1,\dots,x_n]$, not including the variable $x$. We can compute $p_i$ for any $i$ with the help of a summation, namely
\[
   p_i=\sum_{i=0}^d \beta_i f(\alpha_i x_1,\dots,\alpha_i x_n,\alpha_i x)\text{ for some constants }\alpha_i,\beta_i.
\]
We define a transfer function between sets on graphs and polynomials. Given a set $S$ of subgraphs of a graph $G$ we define the given polynomial $\mathcal{P}(S)$ to be
\[
   \mathcal{P}(S) = \sum_{G\in S} \prod_{v\in V(G)} X_v \prod_{e\in E(G)} E_e
\]
 where $X,E$ are variables. Notice, that as expected $\mathcal{P}(C_{n,k})$ where $C_{n,k}$ is the set of all $k$-cliques on a complete graph is exactly our previously defined clique polynomial $\Cl_{n,k}$.

\subsection{Hardness of the $k$-permanent}
We are adapting the proof from Curticapean and Marx~\cite{DBLP:conf/focs/CurticapeanM14} to show the hardness of a parameterized permanent. We give most of the details in this section but will skip correctness proofs unless they are necessary for us.
In the following, we will talk as if we were constructing graphs for clarity of presentation but in reality, we take a large enough $n$ and set variables to zero for edges that we did not construct.

Notice that the theorem from~\cite{DBLP:conf/focs/CurticapeanM14} is not enough for us. It is in general unclear how and in which way the cycles transfer in the theorem while we need an explicit fpt-projection. In essence, we can see this as the following: Our polynomial will contain a set of all $k$-matchings which we need to be able to transfer to a set of $k$-cliques with an fpt-projection. Transferring this will not be difficult but certain care has to be taken in the proofs.

\begin{definition}[Definition 2.4]
   Let $\Gamma$ be a set of colors. A \emph{colored graph} is a graph $G$ together with a coloring $c_G\colon V(G)\rightarrow \Gamma$. We call a graph colorful if $c_G$ is bijective. Two colorful graphs $G, G'$ are \emph{color-preserving} isomorphic if there is a isomorphism $\phi$ from $H$ to $H'$ such that it maps every $\gamma$-colored vertex in $H$ to a $\gamma$-colored vertex in $H'$.
\end{definition}

\begin{definition}[Definition 2.5]\mbox{}
   \begin{itemize}
      \item For colored graphs $H,G$ with $H$ being colorful. We define $\Psub(H\rightarrow G)$ to be the set of all subgraphs $F\subseteq G$ such that $F$ is color preserving isomorphic to $H$.
      \item For a class $\mathcal{H}$ of uncolored graphs, $\Psub(\mathcal{H})$ is the set of all $\Psub(H\rightarrow G)$ where $H$ is vertex colorful graph whose underlying uncolored graph is in $\mathcal{H}$ and $G$ is vertex colored. The parameter is $\lvert V(H)\rvert$.
   \end{itemize}
\end{definition}

Remember that the clique polynomial is defined by $\sum_{C\subseteq [n], \lvert C\rvert =k}\prod_{i,j\in C, i<j} E_{i,j} \prod_{i\in C}X_i$.
\begin{lemma}[Theorem 3.2]\label{lem:partitionsubgrid}
    $\mathcal{P}(\Psub(\mathcal{H}_\text{grid}))$ is $\VW{1}$ hard under fpt-c-reductions.
\end{lemma}
\begin{proof}
   We reduce to the parameterized clique polynomial.

   \begin{itemize}
      \item For every $i\in [k]$ and every $x\in V(G)$ we introduce a vertex $v_{i,i,x,x}$ of color $(i,i)$.
      \item For every $i,j\in [k], i\neq j$ and every $x,y\in V(G)$ such that $x\neq y$ and $\{x,y\}\in E(G)$, we introduce a vertex $v_{i,j,x,y}$ of color $(i,j)$.
      \item For every $i\in [k], j\in [k-1]$ and $x,y,y'\in V(G)$ which are pairwise distinct, if $v_{i,j,x,y}$ and $v_{i,j+1,x,y'}$ both exist in $G'$ then we make them adjacent.
      \item For every $i\in [k], j\in [k-1]$ and $x,x',y\in V(G)$ which are pairwise distinct, if $v_{i,j,x,y}$ and $v_{i,j+1,x',y}$ both exist in $G'$ then we make them adjacent.
   \end{itemize}
   Left to show is the correctness and existence of a projection.
   Let $a_1,\dots,a_k$ be a $k$-clique in the complete graph. This graph is mapped to a colorful grid in our graph $G'$. Now, we can notice that $h_{i,j}$ is mapped to a vertex $v_{i,j,a_i,a_j}\in G'$ and $h_{i,j},h_{i,j+1}$ are adjacent if $v_{i,j,a_i,a_j}$ and $v_{i,j+1,a_i,a_{j+1}}$ exist and are adjacent. Notice that vertices $v_{i,j,a_i,a_j}$ correspond to edges in our clique between $a_i,a_j$ and vertices in our clique correspond to $v_{i,i,a_i,a_i}$ for some $i$. Hence, projecting to these, gives us the required polynomial.

   The correctness of this mapping now follows from~\cite{DBLP:conf/focs/CurticapeanM14}.
\end{proof}

\begin{lemma}[Lemma 3.1]\label{lem:partitionsubminor}
   Let $\mathcal{H}, \mathcal{H'}$ be recursively enumerable graph classes such that for every $H\in \mathcal{H}$ there exists some $H'\in\mathcal{H'}$ with $H\preceq H'$. Then $\mathcal{P}(\Psub(\mathcal{H}))\lefp \mathcal{P}(\Psub(\mathcal{H'}))$.
\end{lemma}
\begin{proof}
   Let us look at $\mathcal{P}(\Psub(\mathcal{H}\rightarrow G))$ and $\mathcal{P}(\Psub(\mathcal{H'}\rightarrow G'))$.
   As in~\cite{DBLP:conf/focs/CurticapeanM14} we can transform $G$ to a $G'$ with larger vertex and edge set. $V(\mathcal{H'})$, the set of vertices of $\mathcal{H'}$, admits a partition into so called branch sets $B_0,\dots,B_k$ such that: For $i\in [1,k]$ the induced graph of $\mathcal{H'}$ on $B_i$ is connected and deleting $B_0$ and contracting each $B_i$ to a single vertex yields some subgraph of $H$ on the vertex set $[k]$.

   Now the following transformation will suffice by the argument in~\cite{DBLP:conf/focs/CurticapeanM14}.
   \begin{enumerate}
      \item For $i\in [1,k]$ and $v\in V_i(G)$, the set of vertices with color $i$: Replace $v$ by the induced subgraph of $\mathcal{H'}$ on $B_i$. Call this graph $L_v$.
      \item For $\{ i,j \}\in E(H)$, the edge set of $H$, and $u\in V_i(G), v\in V_j(G)$: Insert all edges between $L_u$ and $L_v$ in $G'$.
      \item For $\{i, j\}\not\in E(H)$ and $u\in V_i, v\in V_j(G)$: Insert all edges between $L_u$ and $L_v$ in $G'$.
      \item Add a copy of the induced subgraph of $\mathcal{H'}$ on $B_0$ to $G'$, connect it to all other vertices of $G'$.
   \end{enumerate}

   With this, we only need to make sure that for every branch set $L_v$ there is only one designated vertex $v$ for which variable we project down to. As a whole branch set will always be selected this is not a problem. A similar argument is used for the edges. This shows that a simple projection reduction is sufficient.
\end{proof}

\begin{lemma}[Lemma 3.6]\label{lem:partitionsubhcubic}
   $\mathcal{P}(\Psub(\mathcal{H}_\text{bicub}))$ is $\VW{1}$ hard under fpt-c-reductions.
\end{lemma}
\begin{proof}
   By Lemma 3.5~\cite{DBLP:conf/focs/CurticapeanM14} every $H$ appears as the minor of some graph $H'\in\mathcal{H}_\text{bicub}$. Hence, \cref{lem:partitionsubminor} finishes the proof.
\end{proof}

For an edge-colored bipartite graph $G$ over the colors $\Gamma$ and $X\subseteq \Gamma$ be a set of colors. Let $M_X(G)$ be the set of all matchings in $G$ that choose exactly one edge from each color in $X$.
\begin{lemma}[Theorem 4.1]\label{lem:colmatch}
   $\mathcal{P}(M_X(G))$ with parameter $\lvert X\rvert$ is $\VW{1}$ hard under fpt-c-reductions.
\end{lemma}
\begin{proof}
   We want to compute $\mathcal{P}(\Psub(\mathcal{H}_\text{bicub}))$ with oracle calls to the polynomial $\mathcal{P}(M_X(G'))$. Let $H,G$ be $k$-vertex colored graphs with $H$ being 3-regular.
   We build the following graph for $G$:
   \begin{enumerate}
      \item Replace each $v\in V(G)$ by a cycle $C_6$ (a cycle of length 6) on the vertices $w_{v,1}$, $z_{v,1}$, $w_{v,2}$, $z_{v,2}$, $w_{v,3}$, $z_{v,3}$ in this order. The edges of the cycle are colored with ${i}\times [6]$.
      \item For $e\in E(H)$ with $e=\{ u,v \}$ let $i,j\in [3]$ be such that $e$ is the $i$th edge incident with $u$ and the $j$th edge incident with $v$. Replace each edge $\{ u',v'\}\in E(G)$ where $u'$ is $u$-colored and $v'$ is $v$ colored by the edge $\{ w_{u,a}, w_{v,a} \}$ of color $\gamma(e)$.
   \end{enumerate}
   Assume now that for every cycle $C_6$, there are exactly 3 edges outgoing matched and no edges in $C_6$ matched (type $(1,\dots,1)$ in the language of Curticapean and Marx).
   We want to show that every such matching corresponds to an isomorphic graph $H'$ on the set of colors. As in Theorem 4.1~\cite{DBLP:conf/focs/CurticapeanM14} every such matching describes a copy of $H$ and every copy of $H$ describes a unique matching of this type. Notice that the edges and vertices are in an easily projectable relation to each other by the construction.

   For each color $i\in [k]$, we define the independent set $I(v)=\{ w_{v,1}, w_{v,2}, w_{v.3}\}$ and $\mathcal{I}(i) = \bigcup_{v\in V_i(G)} I(v)$ where $V_i(G)$ are all vertices of color $i$ in $G$.
   We call two vertices $u,v\in V(M)\cap \mathcal{I}(i)$ equivalent if there exists some $w\in V(G)$ such that $u,v\in I(w)$. In essence, this translates to that vertices are equivalent if they originate from a single original $w\in G$. Notice that $\lvert\mathcal{I}(i)\rvert=3$. This implies that we get a partition that has at most 5 settings for every color $i$ where the first type is the class of all 3 vertices being equivalent. We call these having type 1 for a color $i$ and having a type vector $(1,\dots,1)$ that describes the types per color.

   It turns out that the matchings are in a bijective relation to $\Psub(H\rightarrow G)$ if we restrict ourselves to matchings having type vector $(1,\dots,1)$. Notice that here we select exactly three vertices for every color $i\in [k]$.
   Now we can use simple extracting of homogeneous components. Every extracting of homogeneous components increases the circuit size by $O(sd^2)$ where $d$ is the degree and $s$ the size. As we have to iterate this, we get a overall bound of $O(sd^{2^k})$ as $d$ is constant in the color variables, we get a total circuit size of $O(sc^k)$.

   Using \cref{lem:partitionsubhcubic} finishes the proof.
\end{proof}

\begin{lemma}[Lemma 2.7]\label{lem:colour}
   Let $M(G)$ be the set of all matchings in a graph $G$. Then $\mathcal{P}(M(G))$ is $\VW{1}$ hard under fpt-c-reductions.
\end{lemma}
\begin{proof}
   \newcommand{\sub}{\ensuremath{\texttt{Sub}}}
   We define $A_S$ as $\sub(H\rightarrow G[E_S])$, i.e., the set of all graphs in $G[E_S]$, the induced graph of the edge set $E_S$, that are isomorphic to $H$. Further notice that for $H$ being $k$ disconnected edges, $\sub(H\rightarrow G)$ are exactly the $k$-matchings in $G$. By the same argument as in Lemma 2.7 in~\cite{DBLP:conf/focs/CurticapeanM14}, we only need to compute the matchings $M$ such that
   \[
      M\in A_X\setminus \bigcup_{S\subsetneq X} A_S.
   \]
   Via an inclusion-exclusion algorithm we can achieve this. Notice, that the monomials, resulting from $A_{S'}$ for any $S'$ are the same as occurring in $A_S$ and hence, we do not need any projection. We need $2^k$ many calls to our oracle.
\end{proof}

We define the $k$-permanent polynomial as follows. Let $S'_n$ be the set of all permutations on $n$ elements that set $n-k$ elements to the identity. Then
\[
   \per_k = \sum_{\sigma\in S'_n} \prod_{i\in [n]} x_{i,\sigma(i)}.
\]
Notice, we do not include the selected vertices, as all vertices are in the cycle cover and hence the two problems are equivalent.
\begin{corollary}
   $(\per_k)$ is $\VW{1}$ hard under fpt-c-reductions.
\end{corollary}
\begin{proof}
   We use the well known identity between matching on bipartite graph and the permanent. Given a bipartite graph (given in a left $L$ and right $R$ side) where we want to find the $k$-matching polynomial. We interpret this graph as a directed graph with edges going from $L$ to $R$. We add self loops to every vertex. Additionally, we connect every vertex on the right side back to all vertices on the left side (in this direction).

   Now every cycle can only form including edges going from $L$ to $R$. We project exactly on these edges and set all other variables to zero to get the required polynomial. The proof is obvious from this simple construction. Note that we need to increase the $k$ to $2k$ to create the correct sized cycles.
   Finishing with \cref{lem:colmatch} and \cref{lem:colour} shows the hardness.
\end{proof}

\subsection{Bounded Length Cycle Covers}
\begin{definition}
   We define \ksparseperm, the bounded length $k$-permanent, to be the polynomial of all cycle covers where one cycle has length $k$ and all other cycles have constant length i.e., length bounded by some constant $c$.
   \[
      \ksparseperm = \sum_{\sigma\in S'_n} \prod_{i\in [n]} x_{i,\sigma(i)}
      \]
      where $S'_n$ is the above mentioned set of permutations on $n$ elements.
   \end{definition}
   By the relationship between permanent and the sum of the weighted cycle cover, we will also speak about weighted bounded $k$-cycle cover.

   \begin{figure}
      \centering
         \begin{tikzpicture}[scale=0.75,->]
            \node[vertex] (a2) at (0,8) {$c_2$};
            \node[vertex] (a) at (0,4) {$c_1$};
            \node[vertex] (b) at (-2,1) {$a$};
            \node[vertex] (c) at (2,1) {$b$};

            \node (u1) at (-4,8) {};
            \node (u2) at (4,8) {};
            \path[dashed] (u1) edge node {} (a2)
            (a2) edge node {} (u2);
            \node (u3) at (-4,4) {};
            \node (u4) at (4,4) {};
            \path[dashed] (u3) edge node {} (a)
            (a) edge node {} (u4);
            \node (u5) at (-4,1) {};
            \node (u6) at (4,1) {};
            \path[dashed] (u5) edge node {} (b);
            \path[dashed] (c) edge node {} (u6);

            \path (b) edge[bend left=20] node [right] {$\frac{1}{2}$} (a);
            \path[thick] (c) edge[bend left=10] node [right] {} (b);
            \path[thick] (a) edge[bend left=20] node [left] {} (b);
            \path[thick] (a) edge[bend right=20] node [right] {} (c);
            \path[thick] (b) edge[bend left=10] node [right] {} (c);
            \path (c) edge[bend right=20] node [left] {$-\frac{1}{2}$} (a);
            \path (c) edge[bend right=50] node [right] {$-\frac{1}{2}$} (a2);
            \path[thick] (a2) edge[bend left] node [right] {} (c);
            \path (b) edge[bend left=50] node [left] {$\frac{1}{2}$} (a2);
            \path[thick] (a2) edge[bend right] node [right] {} (b);
            \path (b) edge[loop below] node [left] {$-1$} (b);
            \path[thick] (c) edge[loop below] node [right] {} (c);
            \path (a) edge[loop above] node [above] {$-\frac{1}{2}$} (a);
            \path (a2) edge[loop above] node [above] {$-\frac{1}{2}$} (a2);
         \end{tikzpicture}
         \caption{A $2$-iff coupling\label{fig:productgadget}}
   \end{figure}
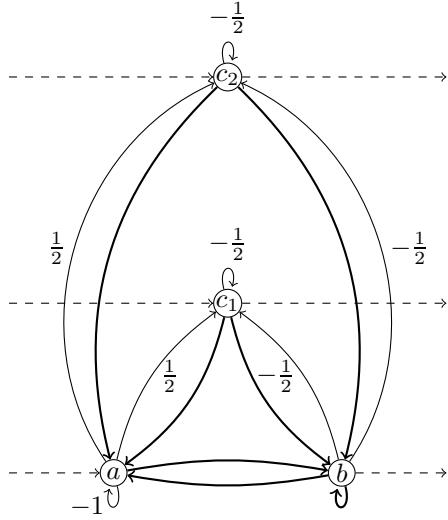
   We call the gadget given in \cref{fig:productgadget} a $2$-iff coupling. The dashed edges allow us to splice this gadget into two given edges of some graph. We will say that the $2$ iff coupling connects these edges. Note that all thick edges have weight $1$. This can be easily extended to a $m$-iff coupling by having vertices $c_1,\dots,c_m$ similarly connected to $a$ and $b$.
   We define the following two notions depending on edges selected in a given cycle cover and not part of the gadget itself.
   \begin{definition}
      We call an $m$-iff coupling \emph{active} if $a$ has an incoming edge, $b$ has an outgoing edge and every $c_i$ has edges $(u,c_i)$ and outgoing edge $(c_i,v)$ where $u\neq v$. We call it \emph{inactive} if all the above edges do not exist and \emph{borderline} otherwise.
   \end{definition}
   We will now proof that we can essentially ignore borderline states if we look at bounded length $k$-cycle covers.
   \begin{lemma}
      Given a directed graph $G$ let $G'$ be a graph where some set of edges $(u,v)$ are connected with a $m$-iff coupling, then the bounded length $k$-cycle cover of $G'$ is equal to the sum of the weighted bounded length $k$-cycle cover on $G$ where in every cycle cover either all edges in the $m$-iff coupling are active or inactive.
   \end{lemma}
   \begin{proof}
      Notice that by proof given by Buergisser~\cite{Buergisser:00}, for every cycle cover having self loops at $c_1,\dots,c_{i-1},c_{i+1},\dots,c_m$ for some $i$ and an arbitrary cover the statement holds. Notice that the extra self-loops only add a constant factor of weight to the cycle cover. By the construction, every cycle cover can have at most one $c_i$ where the self-loop is not taken as otherwise $a$ or $b$ would occur multiple times, The only cases left to prove are the following two.

      We define the gadget being \emph{split} if there exists indices $i,j$ such that $c_i$ is connected in the gadget but $c_j$ is connected outside the gadget, i.e., the edges $(c_i,c_i)$ and $(u,c_j)$ and $(c_j,v)$ for $u\neq v$ are in the cycle cover, $u,v\not\in \{ a,b,c_1,\dots,c_m \}$.
      \begin{description}
         \item[The gadget is split.] If the gadget is split, we can repeat the argument by Buergisser without the vertex $c_j$ as this will just add a constant factor to the sum and hence the lemma is fulfilled.
         \item[The gadget is active.] If the gadget is active, we cannot have any additional edges, hence, the value is $1$.\qedhere
      \end{description}
   \end{proof}

   \begin{definition}
      A \emph{parse tree} $T$ of an arithmetic circuit is the following tree defined iteratively.
      \begin{itemize}
         \item The root is in the tree.
         \item For every summation gate $v$, if  $v\in T$ then exactly one of its children is in $T$.
         \item For every multiplication gate $v$, if $v\in T$ then all of $v$s children are in $T$.
      \end{itemize}
      The weight of a parse tree is the product of all its leafs.
   \end{definition}
   It is a well known result that the sum over the weight of all parse trees of a formula is the polynomial computed by the formula.

With this, we can prove the following theorem.
\begin{theorem}
   For all $t$, there exists a constant $c$ such that $\ksparseperm$ is hard for $\VW{t}$ under fpt-c-reductions.
\end{theorem}
\begin{proof}
   We will show an equivalence between cycle covers of a certain kind and parse trees of the circuit.
   We have a constant depth circuit $C$ for every polynomial in $\VW{t}$ of weft $t$. We will reduce this circuit if necessary to have alternating layers. We also ensure that every path from the root has the same length by adding gates with only one input. This increases the fan-in by at most a constant and the size by a constant factor.

   We convert the circuit first into a tree.  Now we set every edge $\{u,v\}$ in the tree to be an edge in our directed graph $(u,v)$ where the depth of $v$ is smaller than the depth of $u$, i.e., all edges are pointing towards the root.

   We call a vertex $u$ an \emph{additive child of $v$} if there exists a path from $u$ to $v$ in the above tree that does not cross any multiplication gate.

   Let $G$ be a directed clique of size $n$. We split every vertex into an incoming and outgoing vertex. Let us call this graph $G'$. We call the edges between $v_\text{in}$ and $v_\text{out}$ for every vertex the selector edge.
   Now, our circuit has variables $Y_1,\dots,Y_n$ which outgoing edges we couple with $1$-iff couplings to the selector edges of $G'$.

   \begin{figure}
      \centering
      \begin{tikzpicture}
         \begin{scope}[shift={(2,-1)}]
            \node[vertex] (u) at (0,1) {$u$};
            \node[vertex] (v) at (0,0) {$\times$};
            \node[vertex] (v1) at (-1,-1) {$v_1$};
            \node[vertex] (v2) at (1,-1) {$v_2$};
            \draw[->] (v1) -- (v);
            \draw[->] (v2) -- (v);
            \draw[->] (v) -- (u);
            \node at (2,0) {$\Rightarrow$};
         \end{scope}
         \begin{scope}[scale=2,shift={(4,0)}]
            \node[vertex] (u) at (0,1) {$u$};
            \path[->] (u) edge[loop left]  (u);
            \node[vertex] (v) at (0,0) {$\times$};
            \path[->] (v) edge[loop left]  (v);
            \node[vertex] (v1) at (-1,-2) {$v_1$};
            \path[->] (v1) edge[loop left]  (v1);
            \node[vertex] (v1p) at (-1,-1) {$v'_1$};
            \path[->] (v1p) edge[loop left]  (v1p);
            \node[vertex] (v2) at (1,-2) {$v_2$};
            \path[->] (v2) edge[loop right]  (v2);
            \node[vertex] (v2p) at (1,-1) {$v'_2$};
            \path[->] (v2p) edge[loop right]  (v2p);
            \draw[->] (v1) -- (v1p);
            \draw[->] (v2) -- (v2p);

            \draw (-0.75,-0.5) rectangle (0.75,-1.5);
            \node at (0,-1) {$2$-iff coupling};
            \node[vertex] (a) at (0,0.65) {$a$};
            \path[->] (a) edge[loop left]  (a);
            \node[vertex] (b) at (0,0.35) {$b$};
            \path[->] (b) edge[loop left]  (b);
            \draw[->] (v) -- (b);
            \path[->] (b) edge[bend left=10] node [right] {} (a);
            \path[->] (a) edge[bend left=10] node [right] {} (b);
            \draw[->] (a) -- (u);

            \node[vertex] (c1) at (-1,-1.5) {$c_1$};
            \path[->] (c1) edge[loop left]  (c1);
            \node[vertex] (c2) at (1,-1.5) {$c_2$};
            \path[->] (c2) edge[loop right]  (c2);
            \path[->] (u) edge[bend left=35] node[right] {} (v);
         \end{scope}
      \end{tikzpicture}
      \caption{Abbridged Multiplication Node Transform\label{fig:multiplication}}
   \end{figure}
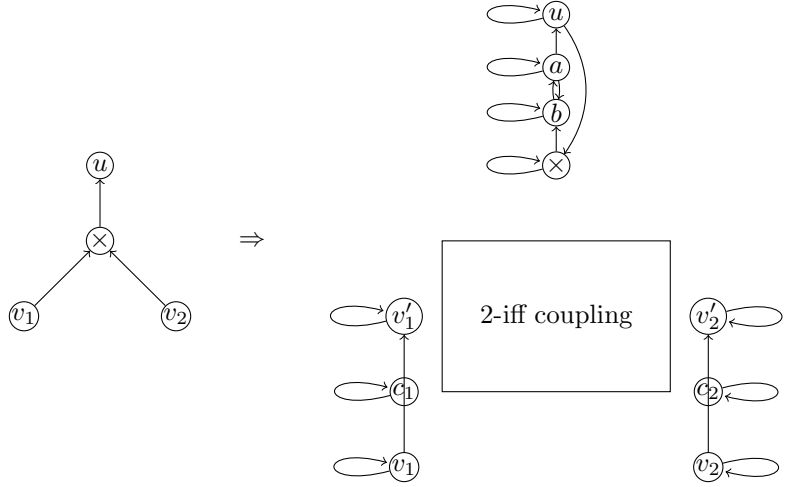

   Now we will construct gadgets around the above given tree. Notice that our paths travel from leafs to roots.
   \begin{description}
      \item[Input Node] We set the weight its outgoing edge to the value in the circuit. The self loop will have weight 1.
      \item[Addition Node] We add a self loop with weight 1.
      \item[Multiplication Node] In essence, we will do the following. We disconnect the graph, construct separate graphs for all its children and enforce that all children will be taken by an appropriate iff coupling. \Cref{fig:multiplication} represents such a transformation with the following caveats. For clarity we omit the weights and the edges belonging to the iff coupling between $(a,c_1)$ and $(a,c_2)$. Additionally, we assume that the vertex $u$ is the last vertex in this tree and omit the loops from $v'-1$, $v'_2$ back to all additive children in their respective subtree.

      In detail:
      For a multiplication gate with fan-in $\ell$ we add a $\ell$-iff coupling. This iff coupling connects the incoming edges with the outgoing edge. Let $a,b,c_1,\dots,c_\ell$ be the vertices from this $\ell$-iff coupling. Let $v$ be the multiplication gate and the edges $(v,u),(v_1,v),\dots,(v_\ell,v)$ be as expected. We splice the $\ell$-iff coupling into this graph in the following way. Add two vertices $a,b$ between $(v,u)$ such that $(a,u)$ and $(v,b)$ are the edges.

      Add new vertices $v'_1,\dots,v'_\ell$, add the edges $(v_i,v'_i)$. Splice between these edges $c_i$, i.e, replace the edge $(v_i,v'_i)$ by $(v_i,c_i), (c_i,v'_i)$. Finally, add for every $v'_i$ edges to all additive children of $v'_i$. We add self-loops to the vertices $v'_i$ of weight 1.
      \item[Root] For the root, we construct outgoing edges to all its additive children in the subtree. If the root is a multiplication gate, we also use the construction for a multiplication node and add an additional vertex in a loop with the root.
   \end{description}
   This finishes the construction of the circuit.
   Finally, we will require via extracting of homogeneous components that the root is always selected, i.e., is not covered by only a self loop, and that $G'$ has $k$ selector edges selected and $2k$ edges overall.

   As the polynomial computes by every formula is equal to the sum over all its parse trees, we just need to show the equality between bounded length $k$-cycle cover and parse trees. Notice that by construction only leafs have weight not equal to one or zero. In the following we will ignore the factor of $k!$ that occurs from the different possible orderings of the edges selected in the clique $G'$. We can easily divide by this factor.

   We first prove that every parse tree has a cycle cover. All vertices not in the parse tree can have self loops with weight 1. All vertices in the path tree are in a cycle. Take the appropriate $k$ cycle in $G'$. With the iff coupling these transfer to certain input gates. Take the cycle starting from the leafs of the parse tree. From this move up along addition gates to a multiplication gate. Here, take the back edge and the iff coupling ensures that we start with a new cycle from the multiplication gates output edge. Finally, we will reach the root. As the depth is constant and every cycle apart from the cycle in the $G'$ is of depth at most constant, this is a valid cycle cover. By construction the weight of this parse tree is the same as the weight of the cycle cover.

   Let us now look at a possible cycle cover. Let us start from the cycle of the root. The cycle goes to the roots children and either hits an addition gate, where it will take exactly one child or it hits a multiplication gate and returns to the root. The multiplication gates iff-coupling ensures that all the children have cycle cover. Repeat this until you have some inputs in the cycle cover. As these are combined with $G'$ by iff coupling, the equivalence follows. By construction the weight of these cycle cover is the same as the weight of the parse tree.
\end{proof}

\section*{Acknowledgements}

We thank Holger Dell for valuable comments on a first draft of this paper
and Radu Curticapean and Marc Roth for helpful discussions.

\bibliography{vftp}
\end{document}